\documentclass{article}
\usepackage[utf8]{inputenc}
\usepackage{amsmath}
\usepackage{amssymb}
\usepackage{amsthm}
\usepackage{physics}
\usepackage{dsfont}
\usepackage{bm}
\usepackage{appendix}
\usepackage{comment}
\usepackage{tablefootnote}
\usepackage{empheq}
\usepackage{hyperref}

\title{Recent mathematical advances in coupled cluster theory}
\author{Fabian M. Faulstich\footnotemark[1]}

\usepackage[dvipsnames]{xcolor}
\usepackage{graphicx}
\usepackage{caption}
\usepackage{subcaption}

\newtheorem{theorem}{Theorem}

\newtheorem{lemma}[theorem]{Lemma}
\newtheorem{proposition}[theorem]{Proposition}

\newtheorem{example}[theorem]{Example}

\usepackage{simpler-wick}
\usepackage{wrapfig}

\begin{document}

\maketitle

\begin{abstract}
This article presents an in-depth educational overview of the latest mathematical developments in coupled cluster (CC) theory, beginning with Schneider's seminal work from 2009 that introduced the first local analysis of CC theory. We offer a tutorial review of second quantization and the CC ansatz, laying the groundwork for understanding the mathematical basis of the theory. This is followed by a detailed exploration of the most recent mathematical advancements in CC theory.
Our review starts with an in-depth look at the local analysis pioneered by Schneider which has since been applied to analyze various CC methods. We then move on to discuss the graph-based framework for CC methods developed by Csirik and Laestadius. This framework provides a comprehensive platform for comparing different CC methods, including multireference approaches.
Next, we delve into the latest numerical analysis results analyzing the single reference CC method developed by Hassan, Maday, and Wang. This very general approach is based on the invertibility of the CC function's Fréchet derivative.
We conclude the article with a discussion on the recent incorporation of algebraic geometry into CC theory, highlighting how this novel and fundamentally different mathematical perspective has furthered our understanding and provides exciting pathways to new computational approaches. 
\end{abstract}

\footnotetext[1]{Department of Mathematical Sciences, Rensselaer Polytechnic Institute, NY, 12180, USA}    
        
\tableofcontents
\newpage 

\section{Introduction}

Coupled-cluster (CC) theory is a widely acclaimed, high-precision wave function approach that is used in computational quantum many-body theory and is of great interest to both practitioners as well as theoreticians~\cite{bartlett2007coupled}. 
The origin of CC theory dates back to 1958 when Coester proposed to use an exponential parametrization of the wave function~\cite{coester1958bound}.
This parametrization was independently derived by Hubbard~\cite{hubbard1957description} and Hugenholtz \cite{hugenholtz1957perturbation} in 1957 as an alternative to summing many-body perturbation theory contributions order by order.
A milestone of CC theory is the work by {\v{C}}{\'i}{\v{z}}ek from 1966~\cite{vcivzek1966correlation}.
In this work, {\v{C}}{\'i}{\v{z}}ek discussed the foundational concepts of second quantization (as applied to many-fermion systems), normal ordering, contractions, Wick’s theorem, normal-ordered Hamiltonians (which was a novelty at that time), Goldstone-style diagrammatic techniques, and the origin of the exponential wave function ansatz.
He moreover derived the connected cluster form of the Schr\"odinger equation and proposed a general recipe for how to obtain the energy and amplitude equations through projections of the connected cluster form of the Schr\"odinger equation on the reference and excited determinants, which was illustrated using the CC doubles (CCD) approximation. 
This work also reported the very first CC computations, using full and linearized forms of CCD, for nitrogen (treated fully at the \emph{ab initio} level) and benzene (treated with a PPP model Hamiltonian).

In this article, we review the most recent mathematical advances from a computational chemistry perspective. Our objective is to elucidate various mathematical frameworks, their objectives, and outcomes in a manner that is accessible to a wide computational chemistry audience. In doing so, we aim to make the complex mathematical concepts accessible to a broader audience, providing a clear and comprehensible pathway for readers who may not have an extensive background in the advanced mathematics typically necessary to fully engage with the original research articles. With this effort, our goal is to render these mathematical results not only understandable but also directly applicable and relevant for practitioners and researchers in the field of computational chemistry.

While this article centers on the mathematical developments in CC theory post-2009, following the landmark work by Schneider, we recognize that there were significant contributions and advancements in the field before this date. However, our focus remains on the period after 2009, showcasing the progress made in recent years. Providing a full account of the rich history of CC theory and the mathematical advances therein is beyond the scope of this article, the interested reader is referred to articles and the references therein that provide insight into the history and development of CC theory including those by Bartlett~\cite{bartlett2005theory}, Paldus~\cite{paldus2005beginnings}, Arponen~\cite{arponen1991independent}, and Bishop~\cite{bishop1991overview}.

The following article is outlined as follows: In Sec.~\ref{sec.MatrixStructure} we provide a brief review of the mathematical matrix structures that arise in the second quantized framework. 
In Sec.~\ref{sec:CCAnsatz} we then introduce the CC ansatz using an algebraic formulation.
In Sec.~\ref{sec:Analysis} we then review the mathematical results established by employing a local strong monotonicity approach (Sec.~\ref{sec:LocalAnalysis}), the excitation graph approach (Sec.~\ref{sec:CCExcitationGraph}) and the inf-sup condition approach (Sec.~\ref{sec:Inf-sup}).
In Sec.~\ref{sec:RootStructure} we then elaborate on the root structure of the CC equations and review the advances made along this line by employing an algebraic geometry perspective.

\section{Brief review of second quantization}
\label{sec.MatrixStructure}

In this section, we review the second quantization framework with a slight mathematical twist. Our aim is to resolve any ambiguities surrounding concepts that have been a potential subject of debate within either the mathematical or chemical community.
Considering an $N$ electron system, we denote the set $\mathcal{B}$ with $|\mathcal{B}|=N_B\gg N$ the set of molecular orbitals, comprising of $L^2$-orthonormal functions, i.e., 
\begin{equation}
\langle
\xi_i,\xi_j
\rangle_{L^2(X)}
=
\int_{X} \xi_i^*(x) \xi_j(x) d\lambda(x)
\qquad \forall~1\leq i,j\leq N_B,
\end{equation}
where $X=\mathbb{R}^3 \times \{ \pm 1/2 \}$ and $d\lambda(x)$ denotes the corresponding integration measure~\cite{fremlin2000measure}. Mathematically, the integral measure $d\lambda(x)$ is a product measure introduced to combine spatial and spin integration, it can also be written as
\begin{equation}
\int_{X} f(\mathbf{x}) d\lambda(\mathbf{x})
=
\sum_{s \in \{\pm 1/2 \}} \int_{\mathbb{R}^3} f(\mathbf{r},s) d{\bf r},
\end{equation}
where $\mathbf{r} \in \mathbb{R}^3$ denotes the spatial degree of freedom and $s\in \{\pm 1/2\}$ is the spin degree of freedom.
We moreover assume that the functions in $\mathcal{B}$ are sufficiently smooth allowing us to take all required derivatives. 
Note that in computations that use Gaussian-type orbitals, this is always the case. 
Mathematically, the largest space (i.e., the most general space) from which we can choose $\mathcal{B}$ is the Sobolev space $H^1(X)$~\cite{leoni2017first,aubin2011applied}, however, for sake of simplicity, one can assume twice continuously differentiable and $L^2$-integrable functions.
In any case, we conclude that the molecular orbitals span a finite-dimensional Hilbert space $h\subset H^1(X)$ which we shall denote the single-particle space. 

Next, we define multi-particle functions that are used to span the fermionic Fock space. 
Due to the anti-symmetry constraints of the wave function, we need to take the anti-symmetrized product also called the {\it exterior product}:
Let $\xi_1,...,\xi_M\in\mathcal{B}$, we define the $M$-folded exterior product of  $\xi_1,...,\xi_M$ (pointwise) by
\begin{equation}
\xi_1\wedge ...\wedge \xi_M (\mathbf{x}_1,...,\mathbf{x}_M)
=
\sum_{\pi \in S_M}
{\rm sign}(\pi)
\prod_{i=1}^M
\xi_{\pi(i)}(\mathbf{x}_i),
\end{equation}
where $S_M$ is the {\it symmetric group} describing all possible permutations of the set $\{1,...,M\}$ and ${\rm sign}(\pi)$ is the parity of the permutation $\pi$.

\begin{example}
Let $\xi_1,\xi_2\in\mathcal{B}$ be two molecular orbitals. The exterior product of $\xi_1$ and $\xi_2$ is pointwise given by 
\begin{equation}
\xi_1\wedge \xi_2 (\mathbf{x}_1, \mathbf{x}_2)
=
\xi_1(\mathbf{x}_1) \xi_2 (\mathbf{x}_2) - \xi_1(\mathbf{x}_2) \xi_2 (\mathbf{x}_1).
\end{equation}
\end{example}

Given the set $\mathcal{B}$, one can form ${N_B \choose M}$ linearly independent exterior products, which define the set $\mathfrak{B}^{(M)}$. 
The space $\mathcal{H}^{(M)}$, spanned by these functions is the {\it $M$-folded exterior power of $h$}, and it inherits an inner product from the single particle space $h$:
Let $\Psi_I =  \xi_{i_1} \wedge ... \wedge \xi_{i_M}$ and $\Psi_J =  \xi_{j_1} \wedge ... \wedge \xi_{j_M}$ then 
\begin{equation}
\langle \Psi_I, \Psi_J \rangle 
= 
\sum_{\substack{ \pi \in S_I\\  \sigma \in S_J}} \prod_{p=1}^M \langle \xi_{\pi(i_p)}, \xi_{\sigma(j_p)} \rangle_{L^2(X)},
\end{equation}
where $S_I$ and $S_J$ are the permutations of $\{i_1,...,i_M\}$ and $\{j_1,...,j_M\}$, respectively.
Normalizing the ${N_B \choose M}$ exterior products obtained from $\mathcal{B}$ by $1/\sqrt{M!}$ yields the well-known definition of $M$-particle Slater determinants, i.e., 
\begin{equation}
\begin{aligned}
\Psi[i_1,...,i_M](\mathbf{x}_1,...,\mathbf{x}_M)
&=
\frac{1}{\sqrt{M!}} 
\sum_{\sigma \in S_M}
{\rm sign}(\pi)
\prod_{i=1}^M
\xi_{\pi(i)}(\mathbf{x}_i)\\
&=
\frac{1}{\sqrt{M!}}
{\rm det} 
\left(
\begin{bmatrix}
\xi_{1}(\mathbf{x}_1) & \cdots & \xi_{1}(\mathbf{x}_M)\\
\vdots & \ddots & \vdots\\
\xi_{M}(\mathbf{x}_1) & \cdots & \xi_{M}(\mathbf{x}_M)
\end{bmatrix}
\right).
\end{aligned}
\end{equation}
To avoid linear dependence in the set of $M$-particle Slater determinants, we assume $i_1<...,<i_M$ which yields ${N_B \choose M}$ possible exterior products formed from $\mathcal{B}$. 
The direct sum of the $M$-particle spaces for $M=0,..., N_B$ yields the fermionic Fock space $\mathcal{F}$:
\begin{equation}
\label{eq:FockSpace}
\mathcal{F} = \bigoplus_{M=0}^{N_B} \mathcal{H}^{(M)},
\end{equation}
which is known as the {\it Grassmann algebra on $h$} in the mathematics community.
For brevity, we will employ the Dirac notation writing the basis elements in $\mathcal{F}$ as
\begin{equation}
|s_1,...,s_{N_B}\rangle = \frac{1}{\sqrt{M!}}
\xi_1^{s_1} \wedge \xi_2^{s_2} \wedge ... \wedge \xi_{N_B}^{s_{N_B}}
\end{equation}
where $M = \sum_{i} s_i$ and $ s_i\in \{0,1\}$ for all $i=1,...,N_{B}$.
A general element in $\mathcal{F}$, is then given as 
\begin{equation}
\label{eq:LinearExpansion}
|\Psi \rangle
=
\sum_{s_1,...,s_{N_B} \in \{ 0, 1\}}
\Psi(s_1,...,s_{N_B}) |s_1,...,s_{N_B} \rangle
\end{equation}
where $\Psi(s_1,...,s_{N_B}) \in \mathbb{C}$. 
We now define the fermionic creation and annihilation operators, i.e., \\
\begin{equation}
\label{eq:FermionicCreation}
\begin{aligned}
a_p^\dag: \mathcal{F} \to \mathcal{F} ~;~ 
|s_1,...,s_{N_B} \rangle
&\mapsto (-1)^{\sigma(p)} (1-s_p) |s_1,...s_{p-1}, 1-s_p, s_{p+1},...,s_{N_B} \rangle\\
a_p: \mathcal{F} \to \mathcal{F} ~;~|s_1,...,s_{N_B} \rangle
&\mapsto (-1)^{\sigma(p)} s_p |s_1,...s_{p-1}, 1-s_p, s_{p+1},...,s_{N_B} \rangle
\end{aligned}
\end{equation}
where $\sigma(p) = \sum_{q=1}^{p-1} s_q$. 
We note that ${\rm dim}(\mathcal{F}) = 2^{N_B}$, we therefore identify elements of the fermionic Fock space $\mathcal{F}$ uniquely with elements in $\mathbb{C}^{2^{N_B}}$. Mathematically, we write $\mathcal{F} \simeq \mathbb{C}^{2^{N_B}}$ which means that the spaces $\mathcal{F}$ and $\mathbb{C}^{2^{N_B}}$ are essentially the same in their structure. We moreover introduce the convention
$$
{1 \choose 0} \equiv {\rm unoccupied } \quad {\rm and} \quad 
{0 \choose 1} \equiv {\rm occupied. } 
$$
Note that this is an arbitrary choice, but it is the commonly employed convention. This allows us to express the basis elements as
\begin{equation}
|s_1,...,s_{N_B} \rangle 
=
{ 1 - s_1 \choose s_1 } \otimes ... \otimes { 1 - s_{N_B} \choose s_{N_B} }. 
\end{equation}

\begin{example}
Let $N_B=2$, then
\begin{equation}
| 0 1 \rangle
=
{1 \choose 0} \otimes {0 \choose 1}
=
\begin{pmatrix}
0\\0\\0\\1
\end{pmatrix}
\end{equation}
\end{example}

In this formulation, the fermionic creation and annihilation operators in Eq.~\eqref{eq:FermionicCreation} are matrices of the form
\begin{equation}
\label{eq:CreationAndAnnihilation}
\begin{aligned}
a_p^\dagger
=
\underbrace{\sigma_z \otimes ... \otimes \sigma_z}_{p-1~{\rm times}} 
\otimes\;
a^\dagger
\otimes 
\underbrace{I \otimes ... \otimes I}_{N_B-p-1~{\rm times}}
\quad 
{\rm and}
\quad
a_p
=
\underbrace{\sigma_z \otimes ... \otimes \sigma_z}_{p-1~{\rm times}} 
\otimes\;
a
\otimes 
\underbrace{I \otimes ... \otimes I}_{N_B-p-1~{\rm times}}
\end{aligned}
\end{equation}
where 
\begin{equation}
I = 
\begin{pmatrix}
1 & 0 \\
0 & 1
\end{pmatrix},\quad
\sigma_z = 
\begin{pmatrix}
1 & 0 \\
0 & -1
\end{pmatrix},\quad
a = 
\begin{pmatrix}
0 & 1 \\
0 & 0
\end{pmatrix}.
\end{equation}

\begin{example}
Let $N_B=3$, then
\begin{equation}
a_2 = 
\sigma_z \otimes a \otimes I
=
\begin{footnotesize}
\begin{pmatrix}
0     &0     &1     &0     &0     &0     &0     &0\\
0     &0     &0     &1     &0     &0     &0     &0\\
0     &0     &0     &0     &0     &0     &0     &0\\
0     &0     &0     &0     &0     &0     &0     &0\\
0     &0     &0     &0     &0     &0    &-1     &0\\
0     &0     &0     &0     &0     &0     &0    &-1\\
0     &0     &0     &0     &0     &0     &0     &0\\
0     &0     &0     &0     &0    & 0     &0     &0\\
\end{pmatrix}
\end{footnotesize}
\end{equation}
\end{example}
These matrices are sparse and have several properties~\cite{helgaker2014molecular}.  
We first note that the definition given in Eq.~\eqref{eq:CreationAndAnnihilation} implies directly that the creation and annihilation operators are nilpotent, i.e., $(a_p^\dagger)^2 = (a_p)^2 = 0$.
Moreover, the fermionic creation and annihilation operators obey the canonical anti-communication relation (CAR):
\begin{equation}
[a_p, a_q]_+ = [a_p^\dag, a_q^\dag]_+ = 0 \quad \text{and} \quad  [a_p, a_q^\dag]_+ = \delta_{p,q}.  
\end{equation}
Lastly, we define the {\it number operator} $n_p = a_p^\dag a_p$ satisfying
\begin{equation}
n_p |s_1,...,s_{N_B} \rangle = s_p |s_1,...,s_{N_B} \rangle
\end{equation}
and the {\it total number operator} $N=\sum_{p=1}^{N_B} n_p$.
In this formulation, the matrix describing an interacting electronic system in a potential generated by clamped nuclei, i.e., the Hamiltonian, takes the form
\begin{equation}
\label{eq:ElecHam}
H
=
\sum_{p,q} h_{p,q} a_p^\dag a_q
+
\frac{1}{4}
\sum_{p,q,r,s} v_{p,q,r,s} a_p^\dag a_q^\dag a_s a_r,
\end{equation}
where $h\in\mathbb{C}^{N_B \times N_B}$ and $v \in \mathbb{C}^{N_B \times N_B \times N_B \times N_B}$ are system dependent integral tensors. They are defined via 
\begin{equation}
h_{p,q} = \int_X \xi_p^*(x) \left(
-\frac{\Delta}{2} - \sum_{j}\frac{Z_j}{|r_1 - R_j|}
\right)\xi_q(x)d\lambda(x)
\end{equation}
and 
\begin{equation}
v_{p,q,r,s} 
=
\int_{X \times X}
\frac{
\xi_p^*(x_1) \xi_q(x_1) \xi_r^*(x_2) \xi_s(x_2)
}{|r_1 - r_2|} d\lambda(x) d\lambda(x).
\end{equation}
Note that $h$ is hermitian and $v$ fulfills the symmetry relations
\begin{equation}
v_{p,q,r,s}
=
v_{r,s,p,q}
=
v_{q,p,s,r}^*
=
v_{s,r, q,p}^*
\end{equation}
or in the case of real-valued atomic spin orbitals
\begin{equation}
v_{p,q,r,s}
=
v_{r,s,p,q}
=
v_{q,p,s,r}
=
v_{s,r, q,p}
=
v_{q,p,r,s}
=
v_{s,r,p,q}
=
v_{p,q,s,r}
=
v_{r,s, q,p}.
\end{equation}
The goal is now to compute the lowest lying eigenstate of the matrix $H$ in the $N$-particle subspace $\mathcal{H}^{(N)}$ which is the $N$-particle ground state energy of the electronic Schr\"odinger equation, i.e., 
\begin{equation}
E_0
=
\min_{
\substack{| \Psi \rangle \in \mathcal{F}\\ 
\langle \Psi | \Psi\rangle = 1}
}
\langle\Psi | H - \mu N |\Psi \rangle, 
\end{equation}
where $\mu$ is a Lagrange multiplier ensuring that the solution $|\Psi \rangle$ lies in the $N$-particle Hilbert space.

\section{The CC ansatz}
\label{sec:CCAnsatz}

Coupled cluster theory is built upon an exponential ansatz of the wave function, as opposed to the linear ansatz of Eq.~\eqref{eq:LinearExpansion}.
We emphasize that the simple approach of projecting the Hamitlonian onto much smaller, manageable linear subspaces of $\mathcal{F}$ proves inadequate for electronic structure problems. This is exemplified in the case study of lithium hydride, where we analyze the lowest eigenstate for different relative positions of $\mathbf{R}_1$ and $\mathbf{R}_2$, denoted $R = \Vert \mathbf{R}_1 - \mathbf{R}_2\Vert$ (see Fig.~\ref{Fig:CCvsCI}). Through this examination, it becomes clear that essential energies, such as the chemical bonding energy, are not accurately captured when using a limited linear subspace, see $E_{\rm bond}^{\rm CISD}$ in Fig.~\ref{Fig:CCvsCI}. Conversely, CC theory provides a far more accurate representation of this critical energy, see $E_{\rm bond}^{\rm CCSD}$ in Fig.~\ref{Fig:CCvsCI}.
\begin{figure}[ht!]
    \centering
    \includegraphics[width=0.5\linewidth]{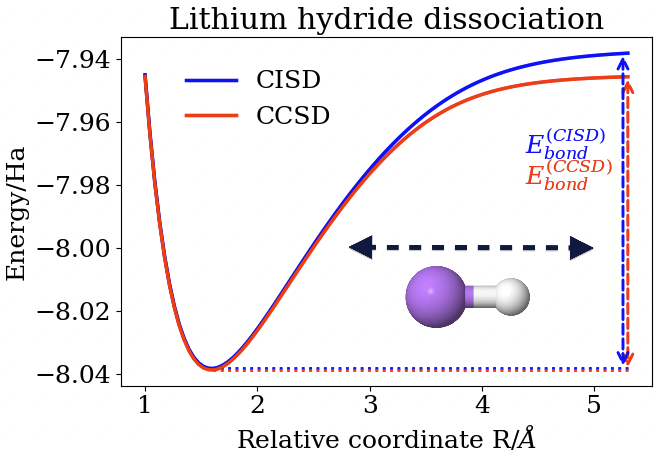}  
    \caption{Case study of lithium hydride comparing the linear parametrization (blue) and the exponential parametrization (red) for different values of $R$ in the AUG-cc-pVTZ basis set~\cite{helgaker2014molecular}.}
    \label{Fig:CCvsCI}
\end{figure}

In order to derive the exponential ansatz in a mathematically sound way, we need to introduce a few concepts first, starting with excitation matrices. 

\subsection{Excitation and cluster matrices}
Since we started the characterization of the fermionic Fock space with the molecular orbitals, the Hartree-Fock state is given by
$$
| \Psi_0 \rangle 
= |1,...,1 ,0,...0\rangle 
= {0 \choose 1} \otimes \cdots \otimes {0 \choose 1} \otimes {1 \choose 0} \otimes \cdots \otimes {1 \choose 0}
\in \mathcal{H}^{N},
$$
where the first $N$ entries are set to one, and the remaining entries are zero. 
We refer to this vector as the reference determinant.
We moreover define $v_{\rm occ} = [\![ N ]\!] = \{1,...,N\}$ and $v_{\rm virt} = [\![ N_B]\!] \setminus [\![ N ]\!] = \{N +1,...,N_B\}$.
Assume $a_1,...,a_k \in v_{\rm virt}$, and $i_1,...,i_k \in v_{\rm occ}$. 
Then, 
$$
X_{a_1,...,a_k \choose i_1,...,i_k}
=
a_{a_k}^\dagger ... a_{a_1}^\dagger a_{i_1} ... a_{i_k}
$$
defines an excitation matrix, and the set of all excitation matrices is given by
$$
\mathfrak{E}(\mathcal{H}^{(N)})
=
\left\lbrace 
X_\mu ~\Big|~ \mu = {a_1,...,a_k \choose i_1,...,i_k},\, a_j \in  v_{\rm virt}, \, i_j \in v_{\rm occ},\, k\leq N
\right\rbrace.
$$
Note that the above construction of the excitation matrices yields that excitation matrices are particle number preserving.
The excitation indices $\mu$ that excite from the occupied into the virtual orbitals define the multi-index set
\begin{equation}
\label{IndexSetI}
\mathcal{I}
=
\left\lbrace
\mu ~\Bigg|~
\mu = {a_1,...,a_k \choose i_1,...,i_k},\, a_j \in  v_{\rm virt}, \, i_j \in v_{\rm occ},\, 1\leq k\leq N
\right\rbrace.
\end{equation}
Since this set of excitations corresponds to simply replacing indices in the string $[1,..., N]$ with indices in the string $[N+1,..., N_B]$ (plus some additional permutation), we deduce that there is a one-to-one relation between excitation operators and Slater determinants except for the reference Slater determinant $|\Psi_0\rangle$. 
In other words, the excitation operators map the reference Slater determinant $|\Psi_0\rangle$ to all other Slater determinants.

\begin{theorem}
\label{th:Basis<->operator}
There exists a one-to-one relation between the $N$-particle basis functions $\mathfrak{B}^{(N)}$ and $\mathfrak{E}(\mathcal{H}^{(N)})\cup\{I\}$.
\end{theorem}
\begin{proof}
Since excitation matrices are defined w.r.t.~the reference determinant $|\Psi_0 \rangle$ it follows immediately that $|\Psi_0\rangle= I |\Psi_0\rangle$.  
Consider $|\Psi_P \rangle = \xi_{p_1} \wedge ... \wedge \xi_{p_N} \in \mathcal{H}^{(N)}$. 
Comparing $\{1,...,N\}$ to $\{p_1,...,p_N\}$ we can identify a multi-index $\mu$ describing the indices that have to be changed in $\{1,...,N\}$ to obtain $\{p_1,...,p_N\}$. 
More precisely, $\mu$ describes an excitation from $v_{\rm occ}\setminus P$ to $P \cap v_{\rm virt}$.
Due to the canonical ordering, this multi-index $\mu$ is unique. 
Then, by definition we obtain $| \Psi_P\rangle = {\rm sign}(\mu) X_\mu |\Psi_0\rangle$, which shows the claim.
\end{proof}
The above result is the fundamental result that allows us to express any target wave function $|\Psi\rangle\in\mathcal{H}^{(N)}$ through a wave operator applied to the reference determinant instead of an expansion through basis vectors, i.e., 
\begin{equation}
\label{eq:waveoperator}
|\Psi \rangle = \left(c_0 I + \sum_\mu c_\mu  X_\mu  \right) |\Psi_0\rangle.
\end{equation}

We will now focus on certain properties that ensure the later discussed exponential formulation is properly defined. The first property we consider is the commutativity of the excitation matrices.

\begin{proposition}
\label{ExcitationMatCommutieren}
Let $X_\mu, X_\nu \in \mathfrak{E}(\mathcal{H}^{(N)})$. Then
$
[X_\mu, X_\nu] = 0
$.
\end{proposition}

\begin{proof}
Let 
$$
X_\mu = X_{a_1,...,a_k \choose i_1,...,i_k}
=
a_{a_k}^\dagger ... a_{a_1}^\dagger a_{i_1} ... a_{i_k}
\quad \text{and} \quad
X_\nu = X_{b_1,...,b_\ell \choose j_1,...,j_\ell}
=
a_{b_\ell}^\dagger ... a_{b_1}^\dagger a_{j_1} ... a_{j_\ell}.
$$
The proof is conducted in two steps:\\
\indent
First, we seek to permute all creation operators in the commutator to the left using the CAR.
We begin with the following product and note that when permuting $a_{b_\ell}^\dagger$ to the right of $a_{a_1}^\dagger$ we merely pick up a sign, since $b_\ell \notin v_{\rm occ}$, i.e., 
$$
\wick{
a_{a_k}^\dagger ... a_{a_1}^\dagger \c1 a_{i_1} ... a_{i_k}
\c1 a_{b_\ell}^\dagger ... a_{b_1}^\dagger a_{j_1} ... a_{j_\ell}
}
=
(-1)^{k}
a_{a_k}^\dagger ... a_{a_1}^\dagger a_{b_{\ell}}^\dagger a_{i_1} ... a_{i_k}
a_{b_{\ell-1}}^\dagger ... a_{b_1}^\dagger a_{j_1} ... a_{j_\ell}.
$$
This furthermore yields
$$
a_{a_k}^\dagger ... a_{a_1}^\dagger  a_{i_1} ... a_{i_k}
 a_{b_\ell}^\dagger ... a_{b_1}^\dagger a_{j_1} ... a_{j_\ell}
=
(-1)^{\ell \cdot k}
a_{a_k}^\dagger ... a_{a_1}^\dagger a_{b_{\ell}}^\dagger ... a_{b_1}^\dagger a_{i_1} ... a_{i_k}
a_{j_1} ... a_{j_\ell}
$$
and similar
$$
a_{b_\ell}^\dagger ... a_{b_1}^\dagger 
a_{j_1} ... a_{j_\ell}
a_{a_k}^\dagger ... a_{a_1}^\dagger  
a_{i_1} ... a_{i_k}
=
(-1)^{\ell \cdot k}
a_{b_\ell}^\dagger ... a_{b_1}^\dagger 
a_{a_1}^\dagger ... a_{a_k}^\dagger
a_{j_1} ... a_{j_\ell} 
a_{i_1} ... a_{i_k}.
$$

Second, we wish to unify the index sequence of the creation and annihilation operators in the two summands of the commutator. 
Applying the CAR again, we find
$$
\wick{
\c1 a_{b_\ell}^\dagger ... a_{b_1}^\dagger  \c1 a_{a_k}^\dagger ... a_{a_1}^\dagger
a_{j_1} ... a_{j_\ell} 
a_{i_1} ... a_{i_k}
}
=
(-1)^{\ell }
a_{a_k}^\dagger a_{b_\ell}^\dagger ... a_{b_1}^\dagger   a_{a_{k-1}}^\dagger ... a_{a_1}^\dagger
a_{j_1} ... a_{j_\ell} 
a_{i_1} ... a_{i_k},
$$
which yields
$$
a_{b_\ell}^\dagger ... a_{b_1}^\dagger  a_{a_k}^\dagger ... a_{a_1}^\dagger
a_{j_1} ... a_{j_\ell} 
a_{i_1} ... a_{i_k}
=
(-1)^{2*\ell \cdot k}
a_{a_k}^\dagger ... a_{a_1}^\dagger
a_{b_\ell}^\dagger ... a_{b_1}^\dagger    
a_{i_1} ... a_{i_k}
a_{j_1} ... a_{j_\ell}.
$$
Note that we have here assumed that $\mu \cap \nu = \emptyset$, otherwise the expression is trivially zero due to the nilpotency of the creation and annihilation operators. 
Overall this yields
\begin{equation}
\begin{aligned}
[X_\mu,& X_\nu] 
= 
[X_{a_1,...,a_k \choose i_1,...,i_k}, X_{b_1,...,b_\ell \choose j_1,...,j_\ell} ]\\
&=
a_{a_k}^\dagger ... a_{a_1}^\dagger 
a_{i_1} ... a_{i_k}
a_{b_\ell}^\dagger ... a_{b_1}^\dagger 
a_{j_1} ... a_{j_\ell}
-
a_{b_\ell}^\dagger ... a_{b_1}^\dagger 
a_{j_1} ... a_{j_\ell}
a_{a_k}^\dagger ... a_{a_1}^\dagger  
a_{i_1} ... a_{i_\ell}\\
&=
(-1)^{\ell \cdot k}
a_{a_k}^\dagger ... a_{a_1}^\dagger a_{b_{\ell}}^\dagger ... a_{b_1}^\dagger 
a_{i_1} ... a_{i_k}
a_{j_1} ... a_{j_\ell}
- 
(-1)^{\ell \cdot k}
a_{b_\ell}^\dagger ... a_{b_1}^\dagger 
a_{a_k}^\dagger ... a_{a_1}^\dagger
a_{j_1} ... a_{j_\ell} 
a_{i_1} ... a_{i_k}\\
&=
(-1)^{\ell \cdot k}
a_{a_k}^\dagger ... a_{a_1}^\dagger a_{b_{\ell}}^\dagger ... a_{b_1}^\dagger 
a_{i_1} ... a_{i_k}
a_{j_1} ... a_{j_\ell}
- 
(-1)^{3\cdot \ell \cdot k}
a_{a_k}^\dagger ... a_{a_1}^\dagger
a_{b_\ell}^\dagger ... a_{b_1}^\dagger    
a_{i_1} ... a_{i_k}
a_{j_1} ... a_{j_\ell}\\
&=0.
\end{aligned}
\end{equation}
\end{proof}
Another important property is that the excitation matrices inherited the nilpotency from the fermionic creation and annihilation matrices.
\begin{proposition}
Let $X_\mu \in \mathfrak{E}(\mathcal{H}^{(N)})$. Then $X_\mu^2 = 0$.
\end{proposition}

\begin{proof}
Recall that $(a_p^\dagger)^2 = (a_p)^2 = 0$ by construction (see Eq.~\eqref{eq:CreationAndAnnihilation}).
Let 
$$
X_\mu = X_{a_1,...,a_k \choose i_1,...,i_k}
=
a_{a_k}^\dagger ... a_{a_1}^\dagger a_{i_k} ... a_{i_1}.
$$
Then
\begin{equation}
\begin{aligned}
X_\mu^2
&= 
\wick{
\c1 a_{a_k}^\dagger ... a_{a_1}^\dagger a_{i_k} ... a_{i_1} \c1 a_{a_k}^\dagger ... a_{a_1}^\dagger a_{i_k} ... a_{i_1}
}\\
&=
- \underbrace{a_{a_k}^\dagger a_{a_k}^\dagger}_{=0} a_{a_k-1}^\dagger... a_{a_1}^\dagger a_{i_k} ... a_{i_1}  a_{a_k-1}^\dagger ... a_{a_1}^\dagger a_{i_k} ... a_{i_1}\\
&=
0.
\end{aligned}
\end{equation}
\end{proof}

We are now set to define the vector space of cluster matrices, a fundamental concept in coupled cluster theory, i.e., the $\mathbb{C}$-vector space
\begin{equation}
\mathfrak{b}
=
\left\lbrace  
T = \sum_{\mu} t_\mu X_\mu ~\Big|~ \mu\in\mathcal{I}
\right\rbrace,
\end{equation}
where $\mathcal{I}$ is as denied in Eq.~\eqref{IndexSetI}.
Note that $\mathfrak{b}$ is a linear space, and the excitation matrices are linearly independent by Theorem~\ref{th:Basis<->operator}, hence, each element in $\mathfrak{b}$ is uniquely defined through its linear coefficients $\mathbf{t} = (t_\mu)$. We therefore use the convention that $\mathbf{t}$ describes an amplitude vector whereas $T$ describes the corresponding cluster matrix.\\  
Utilizing the propositions discussed earlier, we will demonstrate that this vector space possesses a highly structured nature. Our next step is to introduce the concept of the exponential of cluster matrices, which forms a key mathematical bridge between cluster matrices and wave operators. This involves drawing a connection between the Lie algebra, as embodied by the cluster matrices, and the Lie group comprising wave operators, thereby establishing an essential theoretical link in our analysis. To begin this exploration, we first assert that $\mathfrak{b}$ constitutes some form of Lie algebra. As it turns out, this assertion holds true.

\begin{theorem}
The space of cluster matrices $\mathfrak{b}$ equipped with the standard matrix commutator $[\cdot , \cdot ]$ forms a nilpotent Abelian Lie algebra. 
\end{theorem}

\begin{proof}
To show that $\mathfrak{b}$ is a Lie algebra, we need to prove that it (i) is a linear space (which is true by construction), and (ii) that there exists an alternating bilinear map (in this case the standard matrix commutator $[\cdot , \cdot ]$) that satisfies the Jacobi identity, i.e., 
$$
[X,[Y,Z]]+[Y,[Z,X]]+[Z,[X,Y]] = 0 \qquad \forall~X,Y,Z\in \mathfrak{b}.
$$
In order to show (ii), we combine Proposition~\ref{ExcitationMatCommutieren} and the bi-linearity of the matrix commutator. This yields that for two cluster matrices $T_1, T_2\in\mathfrak{b}$
$$
[T_1, T_2] 
=\sum_mu \sum_\nu t_\mu t_\nu [X_\mu, X_\nu]
=\sum_mu \sum_\nu t_\mu t_\nu [X_\nu, X_\mu]
= [T_2, T_1],
$$
hence, cluster matrices commute. 
Therefore, the Jacobi identity is trivially fulfilled. 
This shows that $\mathfrak{b}$ equipped with the regular matrix commutator is an abelian Lie algebra, where the term ``abelian'' simply means that the elements in $\mathfrak{b}$ commute with each other.\\
Next, we shall show the nilpotency. 
To that end, we expand $T^{N+1}$ which yields
\begin{equation}
T^{N+1}
=
\sum_{
\substack{
k_{1}+k_{2}+\cdots +k_{m}=N+1\\ 
k_{1},k_{2},\cdots ,k_{m}\geq 0}}
{N+1 \choose k_{1},k_{2},\ldots ,k_{m}}\prod _{j=1}^{m}(t_{\mu_j} X_{\mu_j})^{k_j},
\end{equation}
where 
$$
{N+1 \choose k_{1},k_{2},\ldots ,k_{m}}
={\frac {N+1!}{k_{1}!\,k_{2}!\cdots k_{m}!}}
$$
is a multinomial coefficient.
Since $|v_{\rm occ}| = N$, there exists one $i\in v_{\rm occ}$ that appears at least twice in each matrix $\prod _{j=1}^{m}(t_{\mu_j} X_{\mu_j})^{k_j}$. 
However, since $a_i^2 = 0 $ this yields that $T^{N+1}=0$, which shows the claim.
\end{proof}

The above Theorem ensures that the (Lie) exponential of $\mathfrak{b}$ is a Lie Group, i.e., a differentiable manifold. However, we seek that this is the Lie Group of wave operators that we used to define any intermediately normalized wave function in $\mathcal{H}^{(N)}$, see Eq.~\eqref{eq:waveoperator}. 
To that end, we begin by showing that every intermediately normalized wave function can be expressed through a linear wave operator. 
As mentioned earlier, the construction of excitation operators allows us to transfer the degrees of freedom from the basis functions in Eq.~\eqref{eq:LinearExpansion} to wave operators. 
Formally, this yields the definition of the (linear) wave operator map $\Omega$ as
\begin{equation}
\Omega : \mathfrak{b} \to \mathcal{G}~;~ C\mapsto I + C,
\end{equation}
where 
\begin{equation}
\mathcal{G} = \{ I + C~|~C\in\mathfrak{b} \}.
\end{equation}
Note that in this formulation, the wave operator map $\Omega$ takes a cluster matrix as input and yields a wave operator, i.e., $\Omega(C)$ maps the reference determinant to some wave function in $\mathcal{H}^{(N)}$.
By construction, $\mathcal{G}$ is an affine linear space of matrices. 
We will now show the one-to-one correspondence between intermediately normalized functions 
\begin{equation}
|\Psi\rangle \in \mathcal{H}_{\rm int} 
= 
\left\lbrace
|\Psi\rangle \in \mathcal{H}^{(N)}~|~\langle \Psi | \Psi_0 \rangle = 1
\right\rbrace
\subset \mathcal{H}^{(N)},    
\end{equation}
and cluster matrices $C\in\mathfrak{b}$. 
We begin with the linear parametrization of elements in $|\Psi\rangle \in \mathcal{H}_{\rm int}$.

\begin{lemma}
\label{lemma:1}
Let $|\Psi\rangle \in\mathcal{H}_{\rm int}$. There exists a unique element $(I +C)\in\mathcal{G}$, s.t., 
\begin{equation}
| \Psi \rangle = (I +C)| \Psi_0\rangle.
\end{equation}
\end{lemma}

\begin{proof}
We first observe that $\mathcal{H}_{\rm int} \subset \mathcal{H}^{(N)}$ can be characterized by 
\begin{equation}
\mathcal{H}_{\rm int} 
= |\Psi_0\rangle + {\rm span}(\{ |\Psi_\mu\rangle \}_\mu)
\underset{(*)}{=} |\Psi_0\rangle + {\rm span}(\{ X_\mu  \}_\mu)|\Psi_0\rangle
= (I + \mathfrak{b})|\Psi_0\rangle,
\end{equation}
where the equality $(*)$ is a consequence of Theorem~\ref{th:Basis<->operator}.
This shows that every element in $|\Psi\rangle \in \mathcal{H}_{\rm int} $ can be expressed as 
\begin{equation}
|\Psi\rangle = (I +C)| \Psi_0\rangle.
\end{equation}
Next, assume there exist two cluster matrices $C_1, C_2 \in \mathfrak{b}$ s.t. 
\begin{equation}
(I +C_1)| \Psi_0\rangle = |\Psi\rangle = (I +C_2)| \Psi_0\rangle.
\end{equation}
However, this yields
\begin{equation}
[c_1]_\mu 
= \langle \Psi_\mu| (I +C_1)| \Psi_0\rangle
= \langle \Psi_\mu| \Psi\rangle
= \langle \Psi_\mu| (I +C_2)| \Psi_0\rangle
= [c_2]_\mu\qquad \forall \mu \in \mathcal{I}
\end{equation}
implying that $C_1 = C_2$, which shows the claim. 
\end{proof}

\begin{lemma}
\label{lemma:2}
The wave operator map $\Omega$ is bijective. 
\end{lemma}
\begin{proof}
First note that $ \mathcal{G}$ was defined by the range of $\Omega$. Hence, the wave operator map is trivially subjective. 
Second, note that ${\rm dim}(\mathfrak{b})= {\rm dim}(\mathcal{G})$, which yields that $\Omega$ is a bijection.
\end{proof}
Combining these two lemmata, yields the desired one-to-one correspondence between $\mathfrak{b}$ and elements in $\mathcal{H}_{\rm int}$.
\begin{theorem}
\label{th:fci}
Let $|\Psi\rangle \in\mathcal{H}_{\rm int}$. There exists a unique element $C\in\mathfrak{b}$, s.t., 
\begin{equation}
|\Psi\rangle = \Omega(C)|\Phi_0\rangle.
\end{equation}
\end{theorem}

Although we have restricted the above theorem to intermediately normalized wave functions (the reason will become apparent shortly), Theorem~\ref{th:fci} is in fact the core of the (full) configuration interaction expansion~\cite{helgaker2014molecular}. 

We now proceed to the exponential parametrization. 
Note, since $\mathfrak{b}$ is nilpotent the exponential series $\exp(T)$ for any element $T\in\mathfrak{b}$ is not a {\it true} exponential as it terminates after at most $N$ terms. Hence, it is a polynomial at most of the degree $N$.
We therefore do not need to investigate the convergence of the exponential series and can define the set
\begin{equation}
\tilde{\mathcal{G}}
=
\left\lbrace
\exp(T)=I + \sum_{n=1}^N\frac{1}{n!}T^n~|~T \in \mathfrak{b}
\right\rbrace.
\end{equation}

\begin{lemma}
\label{lemma:5}
The set $\tilde{\mathcal{G}} $ is equal to $ \mathcal{G}$.
\end{lemma}
\begin{proof}
Let $\exp(T) \in \tilde{\mathcal{G}}$ with $T \in \mathfrak{b}$. By definition
$
\exp(T) = I + P(T)
$
where $P$ is a polynomial at most of the degree $N$ and since $\mathfrak{b}$ is a vector space, we have $P(T)\in \mathfrak{b}$. However, this defines an element in $\mathcal{G}$ which yields $\tilde{\mathcal{G}} \subseteq \mathcal{G}$. 
Conversely, let $I + C \in \mathcal{G}$. Then $I + C - I = C \in \mathfrak{b}$, which implies that 
\begin{equation}
\log (I + C) = \sum_{n=0}^\infty \frac{(-1)^n}{n+1} C^{n+1} 
\end{equation}
terminates after $N+1$ terms. Hence ${\rm log}(I + C)$ is an element in $\mathfrak{b}$ and therewith 
$$
I + C
=
\exp({\rm log}(I + C)) \in \tilde{\mathcal{G}}
$$
which shows that $\mathcal{G}\subseteq \tilde{\mathcal{G}} $.\\
\end{proof}

The common algebraic definition of the Lie exponential map is by means of a map ${\exp : \mathfrak{b} \to \mathcal{G}}$, where $\mathcal{G}$ is a Lie group and $\mathfrak{b}$ the corresponding Lie algebra. 
In particular, the exponential map is a map from the tangent space to the Lie group~\cite{hall2013lie,kirillov2008introduction}. 
We wish to equip $\mathcal{G}$ with a particular group multiplication $\odot$, such that $(\mathcal{G},\odot)$ is a Lie group and $\mathfrak{b}$ its Lie algebra. 
This group multiplication $\odot$ is defined by means of the Backer--Campbell--Hausdorff formula 
$$
\odot: \mathcal{G} \times \mathcal{G}  \rightarrow \mathcal{G};\ \exp(T) \odot \exp(U) = \exp(T * U)
$$ 
for an operation $*$ on $\mathfrak{b}$ which takes the following simple form on Abelian algebras 
\begin{equation}
*: \mathfrak{b} \times \mathfrak{b} \to \mathfrak{b}~;~
(T,S) \mapsto T + S.
\end{equation}
In other words, we can almost trivially derive the coupled cluster ansatz using concepts from non-linear algebra~\cite{michalek2021invitation}.  

\begin{theorem}
\label{Th:ExpMapSurj}
Given the Lie group $\mathcal{G}$ with Lie algebra $\mathfrak{b}$. The exponential map $\exp :\mathfrak{b} \to \mathcal{G}$ is surjective. 
\end{theorem}

Note that this theorem can be generalized to any nilpotent Lie algebra. 
However, the proof shows that the inverse of the exponential is in this particular case well-defined, which proves the following theorem. 

\begin{theorem}
\label{th:ExpMapInv}
The exponential map from $\mathfrak{b}$ to $\mathcal{G}$ is bijective. 
\end{theorem}

This shows that any wave function that is intermediately normalized can be uniquely expressed through an element in $\mathcal{G}$, i.e., through the exponential of a cluster matrix $T\in\mathfrak{b}$. This aligns with the known functional analytic results~\cite{schneider2009analysis,rohwedder2013continuous,laestadius2018analysis,faulstich2019analysis}, and is known in the quantum-chemistry community as the equivalence of FCI and FCC. 

Some of the above results naturally extend to the truncated case, i.e., using a subspace $\bar{\mathfrak{b}} \subset \mathfrak{b}$ in the above construction. We refer the interested reader to \cite{faulstich2022coupled}.

\subsection{The single reference CC theory}

After the mathematical introduction to the CC ansatz, we now turn to the equations that yield the desired cluster matrix. These equations, central to coupled cluster theory, are the {\it coupled cluster equations}. 
This set of equations can be motivated as follows:
Let $|\tilde{\Psi} \rangle \in \mathcal{H}^{(N)}$ be the ground state solution to the electronic Schr\"odinger equation, i.e., 
\begin{equation}
\label{eq:SE}
H|\tilde{\Psi} \rangle 
= 
E_0|\tilde{\Psi}\rangle. 
\end{equation}
We can renormalize $|\tilde{\Psi} \rangle$ to be intermediately normalized, i.e., 
\begin{equation}
|\Psi \rangle  = 
\frac{1}{\langle \Psi_0 |\tilde{\Psi} \rangle }|\tilde{\Psi} \rangle. 
\end{equation}
By Theorem~\ref{th:ExpMapInv}, we then know that there exists a unique element $T\in\mathfrak{b}$ such that 
\begin{equation}
\label{eq:CCansatz}
|\Psi \rangle  = {\rm exp}(T)|\Psi_0 \rangle. 
\end{equation}
Substituting Eq.~\eqref{eq:CCansatz} in the electronic Schr\"odinger equation~\eqref{eq:SE} yields
\begin{equation}
H{\rm exp}(T)|\Psi_0 \rangle
= 
E_0{\rm exp}(T)|\Psi_0 \rangle
\Leftrightarrow
\left\lbrace
\begin{aligned}
\langle \Psi_0 |{\rm exp}(-T)H{\rm exp}(T)|\Psi_0 \rangle &= E_0\\
\langle \Psi |{\rm exp}(-T)H{\rm exp}(T)|\Psi_0 \rangle &= 0,\qquad \forall \langle \Psi | \perp  \langle \Psi_0|.
\end{aligned}
\right.
\end{equation}
Noting that the cluster matrix is a function of the cluster amplitudes ${\bf t}$, i.e., 
\begin{equation}
T({\bf t})
=
\sum_\mu t_\mu X_\mu,
\end{equation}
and that by construction $\langle \Psi_\mu | \Psi_0 \rangle=0$ for all $\mu$, we see that the coupled cluster amplitudes fulfill the square system of equations
\begin{equation}
\label{eq:ccEqs}
0
=\langle \Psi_\mu |{\rm exp}(-T({\bf t}))H{\rm exp}(T({\bf t}))|\Psi_0 \rangle 
=:f_\mu({\bf t})
\qquad \forall \mu.
\end{equation}
Mathematically, CC methods use roots of a high-dimensional non-linear function
\begin{equation}
\label{eq:CCfunction}
f_{\rm CC}: {\bf t} \mapsto [f_\mu({\bf t})]_\mu
\end{equation}
to characterize physical states. 
The above derivation proves that 
\begin{equation}
H|\tilde{\Psi} \rangle 
=
E_0|\tilde{\Psi} \rangle 
\Rightarrow
\left\lbrace
\begin{aligned}
\langle \Psi_0 |{\rm exp}(-T({\bf t}))H{\rm exp}(T({\bf t}))|\Psi_0 \rangle &= E\\
\langle \Psi_\mu |{\rm exp}(-T({\bf t}))H{\rm exp}(T({\bf t}))|\Psi_0 \rangle &= 0\qquad \forall \mu,
\end{aligned}
\right.
\end{equation}
for $T({\bf t})$ fulfilling 
$$
{\rm exp}(T({\bf t}))|\Psi_0 \rangle = \frac{1}{\langle \Psi_0 |\tilde{\Psi} \rangle }|\tilde{\Psi} \rangle.
$$
Note that the converse direction also holds.
Let $| \Psi \rangle \in \mathcal{H}^{(N)}$ be arbitrary and ${\bf t}$ fulfilling Eq.~\eqref{eq:ccEqs}. 
We define
$$
E_{\rm CC}({\bf t}) =  \langle \Psi_0 |{\rm exp}(-T({\bf t}))H{\rm exp}(T({\bf t}))|\Psi_0 \rangle.
$$
Then
\begin{equation}
\label{eq:CCiffRQ}
\begin{aligned}
\langle \Psi | (H - E_{\rm CC}) {\rm exp}(T)|\Psi_0 \rangle
&=
\langle \Psi | {\rm exp}(T){\rm exp}(-T)(H - E_{\rm CC}) {\rm exp}(T)|\Psi_0 \rangle\\
&=
\langle \Psi | {\rm exp}(T)| \Psi_0 \rangle \langle \Psi_0|{\rm exp}(-T)(H - E_{\rm CC}) {\rm exp}(T)|\Psi_0 \rangle\\
&\quad + \sum_\mu 
\langle \Psi | {\rm exp}(T)| \Psi_\mu \rangle \langle \Psi_\mu|{\rm exp}(-T)(H - E_{\rm CC}) {\rm exp}(T)|\Psi_0 \rangle\\
&=0.
\end{aligned}
\end{equation}
Since $| \Psi \rangle \in \mathcal{H}^{(N)}$ was chosen arbitrarily, this shows that ${\rm exp}(T)|\Psi_0 \rangle$ is an eigenvector of $H$ corresponding to the eigenvalue $E_{\rm CC}$.

In practical applications, $|\Psi \rangle $ is of course not known, instead, we seek to find an amplitude vector ${\bf t}$ that fulfills the non-linear equations~\eqref{eq:ccEqs}.
Moreover, we are considering the subspace $\bar{\mathfrak{b}} \subset \mathfrak{b}$ instead of the full space $\mathfrak{b}$.
In order to still obtain a square system of equations, i.e., as many variables as equations, we merely consider the equations that arise from projections that correspond to the excitation matrices used to expand the sought cluster matrix.

It is worth noticing that in this case, the coupled cluster solution is no longer equivalent to the quantum mechanical energy expression. 
In fact, this does -- in general -- not even yield an eigenpair.
This becomes apparent by inspecting Eq.~\eqref{eq:CCiffRQ} and noting that in order to be exactly zero, the CC equations have to contain projections onto all basis functions $\langle \Psi_\mu |$.  

Restrictions to different $\bar{\mathfrak{b}}$ can be motivated from many physical and chemical perspectives, however, mathematically, we consider these restrictions to be sparsity patterns enforced onto the CC amplitude vector ${\bf t}$.
In this context, it is worth noticing that there exists no mathematical result showing the general existence of a sparsity pattern, a sought sparsity pattern is rather the result of computational limitations as well as many computational results indicating that even for complicated systems a certain sparsity in {\bf t} is apparent. 
As such we think of this as a conjecture rather than a fact. 

As a system of nonlinear equations, the equations~\eqref{eq:ccEqs} have a number of solutions. 
Speaking of {\it the} coupled cluster solution bears therefore a certain level of ambiguity. 
Most coupled cluster implementations seek a solution that is close to zero employing a quasi-Newton approach and an initial guess for ${\bf t}$ that stems from MP2. 
Given the convergence behavior of quasi-Newton methods, together with the interesting structures that arise when considering the basins of convergence, this approach seems appropriate for a set of ``well-behaved'' problems but is not a generally applicable procedure. 
This has resulted in a number of numerical advances together with chemically or physically motivated adjustments of the considered system of equations.

\section{Analysis}
\label{sec:Analysis}

The numerical analysis of coupled cluster methods witnessed a significant surge since 2009 when Schneider published the pioneering work that introduced the first local analysis based on Zarantonello's lemma~\cite{zeidler2013nonlinear} to coupled cluster theory. 
This work set the stage for several follow-up works and motivated the exploration of alternative mathematical frameworks well-suited for describing coupled cluster methods.

In Section~\ref{sec:LocalAnalysis}, we outline Schneider's approach and elaborate on the central ideas. We then proceed in Section~\ref{sec:CCExcitationGraph} by introducing the graph-based framework for CC methods developed by Csirik and Laestadius. This perspective introduced novel ideas offering a unified platform to compare various CC methods, including multireference approaches. 
In Section~\ref{sec:Inf-sup} we then elaborate on the most recent numerical analysis results characterizing the single reference CC method. The authors Hassan, Maday, and Wang presented yet another and -- compared to the local analysis -- a more general approach based on the invertibility of the CC Fréchet derivative.\\

Before delving into these analytical characterizations of CC theory, we have to elaborate on three subtle mathematical details that are important to keep in mind when reading this section:

The first is related to the wave function. As outlined earlier, the most general space in which we seek to find a solution to the electronic Schr\"odinger equation is an anti-symmetrized Sobolev space~\cite{schneider2009analysis}. Although we will avoid this detail explicitly in the subsequent elaborations, it is an important detail and a central concept that appears in all analysis works related to CC theory. From a quantum chemistry perspective, seeking a solution within this space ensures finite kinetic energy, in other words:
\begin{equation}
\int_{X^N} |\nabla \psi(\mathbf x_1, \dots,\mathbf x_N )|^2 \mathrm{d}\lambda( \mathbf x_1) \dots \mathrm{d}\lambda( \mathbf x_N) < +\infty.
\end{equation}
For more details on Sobolev spaces, we refer the interested reader to mathematical textbooks~\cite{leoni2017first,aubin2011applied,hackbusch2017elliptic,yserentant2010regularity} or relevant articles that offer insights into their application in quantum chemistry~\cite{laestadius2019coupled,faulstich2020mathematical,rohwedder2010analysis,yserentant2004regularity}. This extra constraint of finite kinetic energy is particularly important for the continuous (i.e., infinite dimensional) formulation of coupled-cluster theory~\cite{rohwedder2013continuous}. 
In this context, we remind the reader of the notation for the $L^2$-inner product $\langle \psi'|\psi\rangle $, and its induced norm $\Vert \psi \Vert^2_{L^2} = \langle \psi |  \psi\rangle$.

The second important detail is a measure of distance on matrix and operator spaces. We here consider operators that act on the wave functions, e.g., the Hamiltonian $H$, cluster operators $T$, $\Lambda$, etc. We can then introduce a norm expression for the operator inherited from the function space it is defined on. For example, let $O$ be an operator defined on $L^2$ then we define the $L^2$ operator norm 
\begin{equation}
\Vert  O \Vert_{L^2}
=
\sup \{ \Vert O\psi \Vert_{L^2} ~:~  \Vert \psi \Vert_{L^2} = 1\}. 
\end{equation}
Note that this concept reduces to induced matrix norms in the finite-dimensional case.

The third detail is that the CC amplitudes ${\bf t}$ live in the Hilbert space of finite square summable sequences denoted the $\ell^2$-space. This space is equipped with the $\ell^2$-inner product~\cite{aubin2011applied}, i.e., let $x=(x_\mu)$ and $y=(y_\mu)$ be two finite sequences, the $\ell^2$-inner product is defined as 
\begin{equation}
\langle x, y \rangle_{\ell^2} = \sum_\mu x_\mu y_\mu,    
\end{equation} 
which induces the norm $\Vert x \Vert^2_{\ell^2} = \langle x, x\rangle_{\ell^2}$.

\subsection{Local strong monotonicity}
\label{sec:LocalAnalysis}

The local analysis introduced by Schneider~\cite{schneider2009analysis} has spawned various works following a similar methodology analyzing different CC methods: Rohwedder generalized it to infinite dimensions~\cite{rohwedder2013continuous,rohwedder2013error}, Laestadius and Kvaal adapted it for the extended CC framework~\cite{laestadius2018analysis}, and Faulstich et al.~adapted it for tailored CC methods~\cite{faulstich2019analysis}. 
Central to all these local analyses is the local version of Zarantonello's lemma~\cite{zeidler2013nonlinear}.

The local version of Zarantonello's lemma states that -- under certain conditions -- a function is (locally) invertible. In the context of coupled cluster theory, the function under investigation is the CC function $f_{\rm CC}$ defined in Eq.~\eqref{eq:CCfunction}. The local invertibility of this function yields the local existence and uniqueness of a CC solution. To ensure the applicability of Zarantonello's lemma, the function in question must exhibit specific characteristics of mathematical ``well-behavedness''. Specifically, in this context, it means that the function must satisfy two essential properties:\\

\emph{Local strong monotonicity}.
The function $f_\mathrm{CC}$ is called locally strongly monotone at $\mathbf{t}_*$ if for some $r>0$, $\gamma>0$ and all $\mathbf{t},\mathbf{t}'$ within the distance $r$ of $\mathbf{t}_*$
\begin{equation}  
\label{eq:delta-def}
\langle f_\mathrm{CC}(\mathbf{t}) - f_\mathrm{CC}(\mathbf{t}'), \mathbf{t} - \mathbf{t}' \rangle_{\ell^2} \geq \gamma \Vert \mathbf{t}-\mathbf{t}' \Vert_{\ell^2}^2.
\end{equation}

\emph{Local Lipschitz continuity}.  
The function $f_\mathrm{CC}$ is said to be locally Lipschitz continuous at $\mathbf{t}_*$ with Lipschitz constant $L > 0$ if for some $r>0$ and all $\mathbf{t},\mathbf{t}'$ within the distance $r$ of $\mathbf{t}_*$
\begin{equation}
    \Vert f_\mathrm{CC}(\mathbf{t}) - f_\mathrm{CC}(\mathbf{t}') \Vert_{\ell^2} \leq L \Vert \mathbf{t} - \mathbf{t}' \Vert_{\ell^2},
\end{equation} 
Note that in the finite-dimensional case, $f_\mathrm{CC}$ is indeed locally Lipschitz since it is continuously differentiable.\\

In the context of CC theory, the difficult property to prove is that $f_{\rm CC}$ -- or the respective function describing the CC method under consideration -- is locally strongly monotone. Here, all analyses generally follow a similar pattern:
Inspecting the left-hand side in Eq.~\eqref{eq:delta-def} we  begin by a tailor expansion of $f_{\rm CC}$ around $\mathbf{t}_*$, i.e.,
\begin{equation}
f_\mathrm{CC}(\mathbf{t}) - f_\mathrm{CC}(\mathbf{t}')
= Df_\mathrm{CC}(\mathbf{t}_*) (\mathbf{t} - \mathbf{t}') + \mathcal{O}\left((\mathbf{t}-\mathbf{t}')^2\right),
\end{equation}
where $Df_\mathrm{CC}$ is the Jacobian of $f_{\rm CC}$. This yields
\begin{equation}  
\langle f_\mathrm{CC}(\mathbf{t}) - f_\mathrm{CC}(\mathbf{t}'), \mathbf{t} - \mathbf{t}' \rangle_{\ell^2} 
=
\langle  Df_\mathrm{CC}(\mathbf{t}_*) (\mathbf{t} - \mathbf{t}'), \mathbf{t} - \mathbf{t}' \rangle_{\ell^2} 
+
\mathcal{O} \left(\Vert \mathbf{t}-\mathbf{t}'\Vert^3\right).
\end{equation}
At this point, it is common to impose certain locality assumptions, i.e., assuming that $\mathbf{t}$ and $\mathbf{t}'$ are close enough to $\mathbf{t}_*$. This ensures that the term $\mathcal{O} \left(\Vert \mathbf{t}-\mathbf{t}'\Vert^3\right)$ is sufficiently small. In order to control the remaining term, the derivative of $f_{\rm CC}$ can explicitly be computed, which yields
\begin{equation}
\langle  Df_\mathrm{CC}(\mathbf{t}_*) (\mathbf{t} - \mathbf{t}'), \mathbf{t} - \mathbf{t}' \rangle_{\ell^2} 
=
\langle  (T-T')\Psi_0, e^{-T_*}\left(H - E_{\rm CC}({\bf \mathbf{t}_*})\right)e^{T_*} (T-T') \Psi_0\rangle_{L^2}.
\end{equation}
The next step involves expanding the similarity-transformed Hamiltonian, i.e.,  $e^{-T_*}H e^{T_*}$, using the Hausdorff Lemma~\cite{hall2013lie}, which is an important lemma derived from the Baker--Campbell--Hausdorff formula. This yields
\begin{equation}
e^{-T_*}\left(H - E_{\rm CC}({\bf t_*})\right)e^{T_*}
=
\left(H - E_{\rm CC}({\bf t_*})\right)
-T_*\left(H - E_{\rm CC}({\bf t_*})\right)
+\left(H - E_{\rm CC}({\bf t_*})\right)T_*
+...
\end{equation}
Again, imposing locality of $t$ and $t'$ around $t_*$ will ensure that higher-order terms become negligible. This yields that 
\begin{equation}
\begin{aligned}
\langle  (T-T')\Psi_0,& e^{-T_*}\left(H - E_{\rm CC}({\bf t_*})\right)e^{T_*} (T-T') \Psi_0\rangle_{L^2}\\
&=
\langle  (T-T')\Psi_0, \left(H - E_{\rm CC}({\bf t_*})\right) (T-T') \Psi_0\rangle_{L^2}\\
&\quad +
\langle  (T-T')\Psi_0, \left(H - E_{\rm CC}({\bf t_*})\right)(T_* - T_*^\dagger) (T-T') \Psi_0\rangle_{L^2}\\
&\quad+
...\\
\end{aligned}
\end{equation}
The first term in this expansion, i.e., 
\begin{equation}
\langle  (T-T')\Psi_0, \left(H - E_{\rm CC}({\bf t_*})\right) (T-T') \Psi_0\rangle_{L^2},
\end{equation}
can then be bounded by imposing different spectral gap assumptions depending on the CC method under consideration. 
The applicability of different spectral gap assumptions is reasonable in the context of CC methods. 
For the remainder term, 
\begin{equation}
\langle  (T-T')\Psi_0, \left(H - E_{\rm CC}({\bf t_*})\right)(T_* - T_*^\dagger) (T-T') \Psi_0\rangle_{L^2}
\end{equation}
further ``well-behavedness'' assumptions have to be made which commonly involve the fluctuation potential $W = H -F$, where $F$ is the Fock matrix. Opposed to the spectral gap assumption there is much less known about the feasibility of such assumptions. 
Combining these estimates yields an approximate strong monotonicity constant denoted by $\Gamma$.
The positivity of this constant varies depending on the system being analyzed, indicating that such an analysis is not universally applicable. For specific values of $\Gamma$ across different systems, we refer to Table~\ref{tab:SMPvsSUPINF} where these variations are detailed.
Moreover, {\it a prior} error estimates can be established using the general framework introduced by Bangerth and Rannacher~\cite{bangerth2003adaptive}.

\subsection{Excitation graphs and topological degree}
\label{sec:CCExcitationGraph}

In a series of two articles~\cite{csirik2023coupled1,csirik2023coupled2} Csirik and Laestadius propose a novel and comprehensive mathematical framework for
Coupled-Cluster-type methods. 

In the first article of this series~\cite{csirik2023coupled1}, the authors develop a graph theoretical approach offering a new interpretation of the excitation structures in various CC methods through a graph-based framework. This method is particularly potent as it enables a cohesive analysis of both single and multi-reference CC methods within a unified structure. To illustrate this concept, we consider a simplified scenario with five spin orbitals, labeled $\{1,...,5\}$, and two reference states, $\{1,2,3\}$ and $\{1,2,4\}$. The array of possible excitations in this setup can be effectively represented using a graph, where each edge symbolizes a potential excitation. This graphical representation is detailed in Figure~\ref{fig:CC-graph}, providing a clear and structured visualization of the excitation dynamics.

\begin{figure}[h!]
    \centering
    \includegraphics[width = 0.8\textwidth]{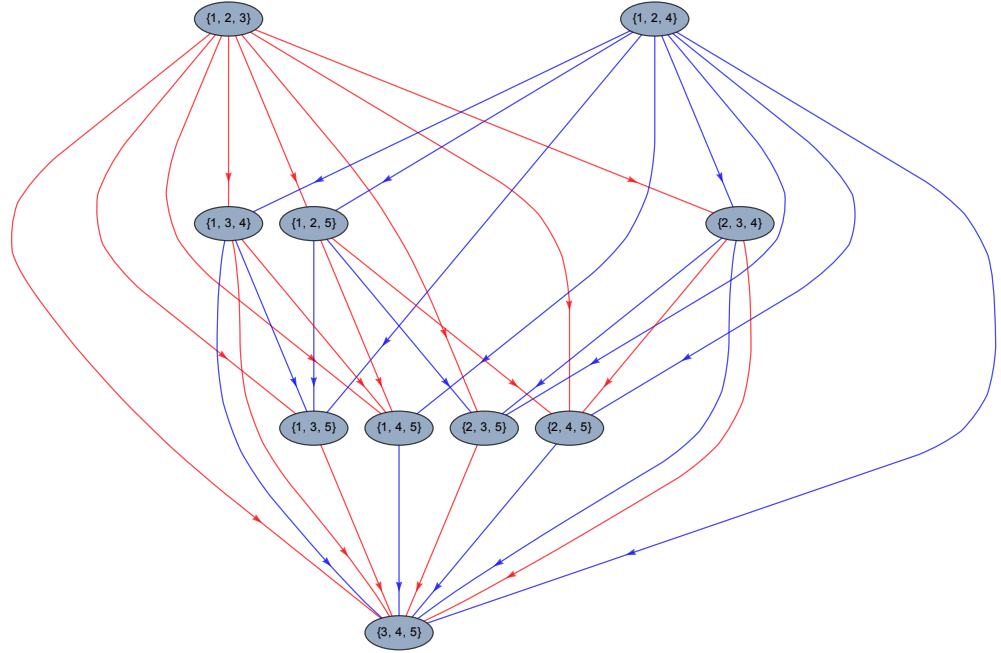}
    \caption{Full multi-reference excitation multigraph for five spin orbitals, labeled $\{1,...,5\}$  and two reference states, $\{1,2,3\}$ and $\{1,2,4\}$. The excitations w.r.t the references $\{1,2,3\}$ and $\{1,2,4\}$ are represented as edges. To distinguish the excitations from $\{1,2,3\}$ and $\{1,2,4\}$, the edges are colored in red and blue, respectively. See~\cite{csirik2023coupled1} for more details.}
    \label{fig:CC-graph}
\end{figure}

Analyzing the excitation graph itself can lead to insights about the considered CC method. As an example, the transitivity of the graph implies the algebraic closedness of the set of excitation operators. 

In the second article of this series~\cite{csirik2023coupled2}, the authors analyze the nonlinear equations arising in the single reference CC method using topological degree theory. This mathematical tool is instrumental in decoding and resolving specific equation types that entail mappings between topological spaces. When applied to the CC map, topological degree theory allows for the deduction of local existence and uniqueness of the CC solutions. Additionally, it facilitates the extraction of the topological index for solutions within the single reference CC framework. In general, the topological index of a root in a nonlinear map is particularly enlightening, shedding light on the root's inherent nature, especially regarding its stability and the map's behavior in its vicinity. In this context, the authors successfully demonstrate the application of topological index results to both non-degenerate and degenerate solutions in the single reference CC method, providing deeper insights into the underlying mathematical structure of these solutions.

In addition to their exploration of nonlinear equations in the single reference CC method, the authors also investigate the complex issue of discerning the ``physicality'' of solutions to truncated CC equations. This area of research has been pivotal in distinguishing between ``physical'' solutions, which accurately mirror real-world phenomena, and ``unphysical'' ones, which are considered irrelevant or misleading. A landmark study by Kowalski and Piecuch~\cite{piecuch2000search} played a crucial role in this context, employing a specific homotopy method to categorize these solutions. Despite some debate over the universality of this method, as noted by Csirik and Laestadius in Remark 4.30 in~\cite{csirik2023coupled2}, the contributions of Kowalski and Piecuch were significant -- to the extent that the authors christened this particular homotopy the ``Kowalski--Piecuch homotopy'', or ``KP homotopy'' for short. Unveiled the intricate nature of solutions to truncated CC equations, results in~\cite{piecuch2000search} highlighted the need for deeper analytical scrutiny.
This revelation has spurred further examination of the CC equations and the KP homotopy approach, with a renewed focus on employing topological degree theory. By doing so, Csirik and Laestadius have markedly enhanced our comprehension of the complex nature inherent in truncated CC equations, offering new perspectives and deeper insights into their behavior and implications.

\subsection{Inf-Sup condition}
\label{sec:Inf-sup}

In their two-part series of articles~\cite{hassan2023analysis,hassan2023analysis2}, Hassan, Maday, and Wang have made substantial advancements in our analytical grasp of the CC function $f_{\rm CC}$ defined in Eq.~\eqref{eq:CCfunction}.
To avoid an {\it ad hoc} bound onto the fluctuation potential as imposed in Sec.~\ref{sec:LocalAnalysis}, the authors instead prove the local invertibility of the CC function through a classical inf-sup type argument that marks a significant shift in the analytical methodology employed in CC theory.
Such an inf-sup condition, also called the Babuška–Brezzi condition which is a technique commonly used when analyzing indefinite elliptic partial differential equations, can be summarized as follows:

Consider the bounded linear mapping $A$ between two normed spaces $(V,\Vert \cdot \Vert_V)$ and $(W,\Vert \cdot \Vert_W)$ -- note that the Jacobian naturally fulfills these assumptions. 
The Babuška–Brezzi condition states that there exists a constant $\alpha > 0$ such that
\begin{equation}
\label{eq:inf-sup}
\inf_{\substack{v\in V\\ v\neq 0}}
\sup_{\substack{w \in W \\ w\neq 0}}
\frac{|A(v,w)|}{\Vert v \Vert_V \Vert w \Vert_W} \geq 0 
\quad
{\rm and}
\quad
\inf_{\substack{w \in W \\ w\neq 0}}
\sup_{\substack{v\in V\\ v\neq 0}}
\frac{|A(v,w)|}{\Vert v \Vert_V \Vert w \Vert_W} \geq \alpha ,
\end{equation}
see~\cite{sauter2011boundary} for more details. 
This condition ensures that the operator $A$ is neither ``too weak'' nor ``too strong'', in the sense that it maps elements of $V$ and $W$ in a balanced way. 
Gaining a clearer understanding becomes easier in the context of finite dimensions: In this scenario, striving for local strong monotonicity, as detailed in Sec.~\ref{sec:LocalAnalysis}, is analogous to verifying that a matrix is positive definite. 
Similarly, the inf-sup condition, as described in this section, can be likened to establishing a matrix's invertibility.
Note that in the realm of finite dimensions, a square matrix's invertibility can be deduced solely from its injectivity. However, the infinite-dimensional scenario demands a bit more nuance. This is why we see two distinct conditions in Eq.~\eqref{eq:inf-sup}, reflecting the additional complexity inherent in infinite dimensions.

In connection with single reference coupled cluster theory, the authors establish this condition for the similarity transformed shifted Hamiltonian, which arises from the coupled cluster Jacobian, see~\cite{hassan2023analysis}. As shown subsequently, proving local invertibility based on this inf-sup condition yields more generally applicable well-posedness results compared to the local analysis techniques described in Sec.~\ref{sec:LocalAnalysis}. 

\subsubsection{Overview of the inf-sup type argument}
The essence of the analysis presented in~\cite{hassan2023analysis} is that the CC function can be locally inverted if and only if its Jacobian, referred to as the CC Jacobian, can be locally inverted. Recall that the CC Jacobian is given by
\begin{equation}
\langle w, Df_{\rm CC}(t) v \rangle
=
\langle
W \Psi_0, e^{T(t)} [H,S] e^{T(t)} \Psi_0
\rangle
\end{equation}
and we introduce the map $A$ via the description
\begin{equation}
\label{eq:CCJac}
\langle
W \Psi_0, A(t) S \Psi_0
\rangle
=
\langle
W \Psi_0, e^{T(t)} [H,S] e^{T(t)} \Psi_0
\rangle
.
\end{equation}

Note that the CC Jacobian $Df$ at ${\bf t}$ is then invertible if and only if $A$ is invertible at ${\bf t}$. The authors leverage this observation and work with $A$ instead of the Jacobian $Df$. Moreover, $A({\bf t}_*)$ is a similarity transformed of the shifted Hamiltonian, in particular, it is non-symmetric! Therefore, one can either study $A({\bf t})$ or $A^\dagger({\bf t})$, both approaches are equivalent, yet one approach might be simpler than the other. Indeed, the authors establish the following two key results which yield the invertibility of $A$ at ${\bf t}_*$ which then yields the invertibility of the CC Jacobian $Df_{\rm CC}$ at ${\bf t}_*$, see Theorem 31 in~\cite{hassan2023analysis}. 
First, the authors prove that at the true, untruncated CC solution ${\bf t}_*$, the function $A({\bf t}_*)$ is injective, see step one in the proof of Theorem 31 in~\cite{hassan2023analysis}. This is equivalent to the first inequality in Eq.~\eqref{eq:inf-sup}. Second, the authors establish that $A^\dagger({\bf t}_*)$ is bounded below, see step two in the proof of Theorem 31 in~\cite{hassan2023analysis}.
This is equivalent to the second inequality in Eq.~\eqref{eq:inf-sup}. Combining these results yields that $A$ at ${\bf t}_*$ is invertible.  A direct consequence of this is that the CC Jacobian $Df$ evaluated at ${\bf t}_*$ is invertible, and its inverse is bounded, i.e., 
\begin{equation}
\label{eq:errorEq}
\Vert 
Df^{-1}({\bf t}_*)
\Vert \leq \frac{\Theta}{\Upsilon},
\qquad
{\rm with}\qquad 
\Theta
=
\Vert e^{T^\dagger({\bf t}_*)}\Vert \Vert \mathbb{P}_0^\perp e^{-T({\bf t}_*)} \Vert,
\end{equation}
where $\Upsilon$ is the inf-sup constant from Eq.~\eqref{eq:inf-sup}, and  $\mathbb{P}_0^\perp$ is the projection onto the space orthogonal to ${\rm span}(\Psi_0)$.

These results, in turn, can then be leveraged to establish that the CC function $f_{\rm CC}$, under some assumptions (see Theorem 33 in~\cite{hassan2023analysis}), is locally invertible around ${\bf t}_*$ and $f_{\rm CC}$ as well as its local inverse are differentiable -- in mathematical parlance, $f_{\rm CC}$ is a local diffeomorphism.  
Moreover, the authors establish a local error bound of the form
\begin{equation}
\Vert {\bf t}_* - {\bf t} \Vert \leq 2 \frac{\Theta}{\Upsilon} \Vert f_{\rm CC}({\bf t}_*)\Vert .
\end{equation}

\subsubsection{Interpretation and results of the inf-sup argument}
Similarly to the local analysis results elaborated on in Sec.~\ref{sec:LocalAnalysis}, the inf-sup argument relies on the positivity of the constants involved. 
The advantage of the analysis presented in~\cite{hassan2023analysis}, is that the constants are provably positive and therefore universally applicable. In particular, they do not rely on assumptions on the fluctuation potential. See Table~\ref{tab:SMPvsSUPINF} for some molecular test systems in equilibrium geometry, and Fig.~\ref{fig:HydrogenFluoride} for bond-dissociation of hydrogen fluoride.  

\begin{table}[h!]
    \centering
    \begin{tabular}{c|c|c|c}
Molecule 
& $1/\Vert Df^{-1}(t_*)\Vert$
& $\Upsilon/\Theta$ 
& $\Gamma$  \\
\hline
BeH$_2$ &0.3379  & 0.2568 & 0.0363 \\
BH$_3$  &0.3060  & 0.2081 & -0.0950 \\
HF      &0.2995  & 0.2529 & -0.0083 \\
H$_2$O  &0.3576  & 0.2789 & 0.0249 \\
LiH     &0.2628  & 0.2164 & -0.0065 \\
NH$_3$  &0.4113  & 0.2784 & -0.0325 \\
\end{tabular}
\caption{
Comparison of the approximate strong monotonicity constant $\Gamma$ (see Sec.~\ref{sec:LocalAnalysis}) and $\Upsilon/\Theta$ (see Eq.~\eqref{eq:errorEq}) -- both seeking a positive lower bound to $1/\Vert Df^{-1}(t_*)\Vert$. The calculations were performed in STO-6G basis sets except for the HF and LiH molecules
for which the 6-31G basis sets was
used. For more details see~\cite{hassan2023analysis}}.
\label{tab:SMPvsSUPINF}
\end{table}

\begin{figure}[h!]
    \centering
    \includegraphics[width=0.8\textwidth]{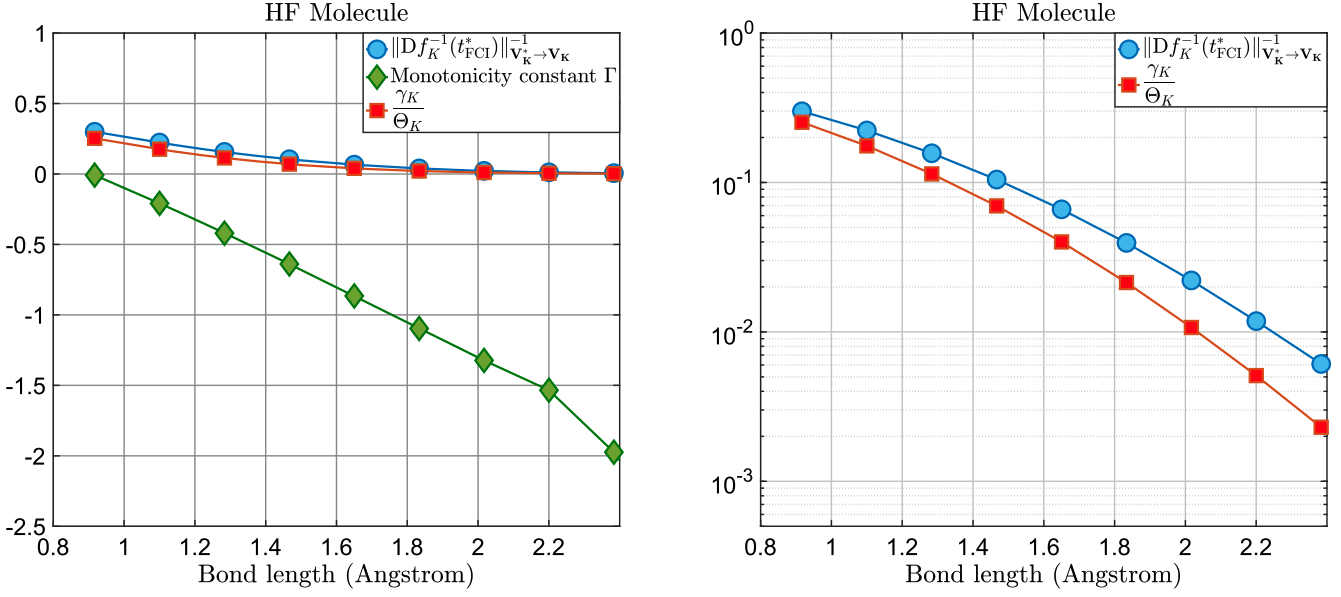}
    \caption{Numerically computed constants for the HF molecule at different bond lengths. The equilibrium
bond length is 0.9168 \AA. The figure on the right uses a log scale on the y-axis. For more details see~\cite{hassan2023analysis}}.
\label{fig:HydrogenFluoride}
\end{figure}

While the analytical approach outlined in~\cite{hassan2023analysis} encounters certain challenges, its contributions to the mathematical understanding of CC theory are pivotal. 
Initially, this analysis seemed limited to the untruncated CC framework, relating approximate untruncated CC solutions with the infinite-dimensional untruncated CC solutions. 
However, the authors adeptly addressed this in a subsequent publication~\cite{hassan2023analysis2}, successfully extending their findings to truncated CC methods. Another complexity lies in the computation of the involved constants in a numerically tractable manner. The constants involve operator norms which are in general not easily accessible, to say the least. 
Moreover, these constants are further linked to either the specific value of the untruncated CC solution ${\bf t}_*$ or the spectral properties of related operators. 
Despite this, the potential for practical application remains promising. Future work could focus on developing manageable approximations of these constants, thereby making the insights from~\cite{hassan2023analysis,hassan2023analysis2} more accessible for practical simulations.

In summary, this novel analytical approach has significantly advanced our understanding of the local behavior of the CC function. It introduces a sound mathematical framework for understanding its local behavior, thereby greatly enriching our knowledge in this area.

\section{The root structure of CC theory}
\label{sec:RootStructure}

As outlined in~\cite{faulstich2022coupled,faulstich2023homotopy}, the root structure of a polynomial system is (in general) of fundamental importance. 
It unveils key aspects~\cite{hubbard2001find}, such as the multiplicity of roots and the nature of these roots, e.g., whether they are real or complex.
Such insights are especially vital when employing (approximate) root-finding methods in practical applications.
The pursuit of roots to the CC equations~\eqref{eq:ccEqs} is a direct application of these principles.

In the context of CC methods, most commonly, (quasi) Newton-type methods are employed to approximate {\it one} root of the CC equations. From a computational perspective, (quasi) Newton-type methods have better numerical scaling than more general root-finding procedures. Additionally, in a perturbative regime, one could argue that the CC amplitudes can be viewed as minor corrections to the HF solution. Consequently, it may be sufficient to approximate a single root near zero, which represents a small change to the HF solution. This perturbation theoretical reasoning is of paramount importance in understanding the current computational and theoretical practices in CC theory. In particular, it justifies the use (quasi) Newton-type methods and also explains the quantum chemical rule of thumb, namely, ``Do not trust simulations with large CC amplitudes''. However, it is very important to note that: \\

\noindent
This reasoning does not cover all cases where CC theory can be successfully applied!~\cite{bulik2015can,giner2018interplay,faulstich2023S_diag}\\

The rule of thumb may be fine in the regime of weakly correlated systems, but it certainly breaks down for strongly correlated systems~\cite{faulstich2022coupled,faulstich2023S_diag}. For strongly correlated systems, it is common practice to make a case-by-case assessment of the computed results, currently limiting the reliable out-of-the-box application of CC methods.
To illustrate the limitations of the perturbation theoretical perspective in fully comprehending CC theory, consider the single polynomial $p(z) = z^3-1$, which has three distinct roots: $z_1 = 1$, and $z_{2,3} = 1/2 \pm i\sqrt{3}/2$. 
Applying Newton's method to approximate {\it one} root to this system, we notice that, depending on the initialization, a different solution is found. This can be visualized by sampling a feasible region in $\mathbb{C}$ and using these points as initialization for Newton's method. Depending on which root was approximated, we then color each point accordingly. This yields the known Newton fractal corresponding to $p(z)$, see Fig.~\ref{fig:Motivation} (left panel).

\begin{figure}[h!]
     \centering
     \begin{subfigure}[b]{0.495\textwidth}
    \centering
    \includegraphics[width = \textwidth]{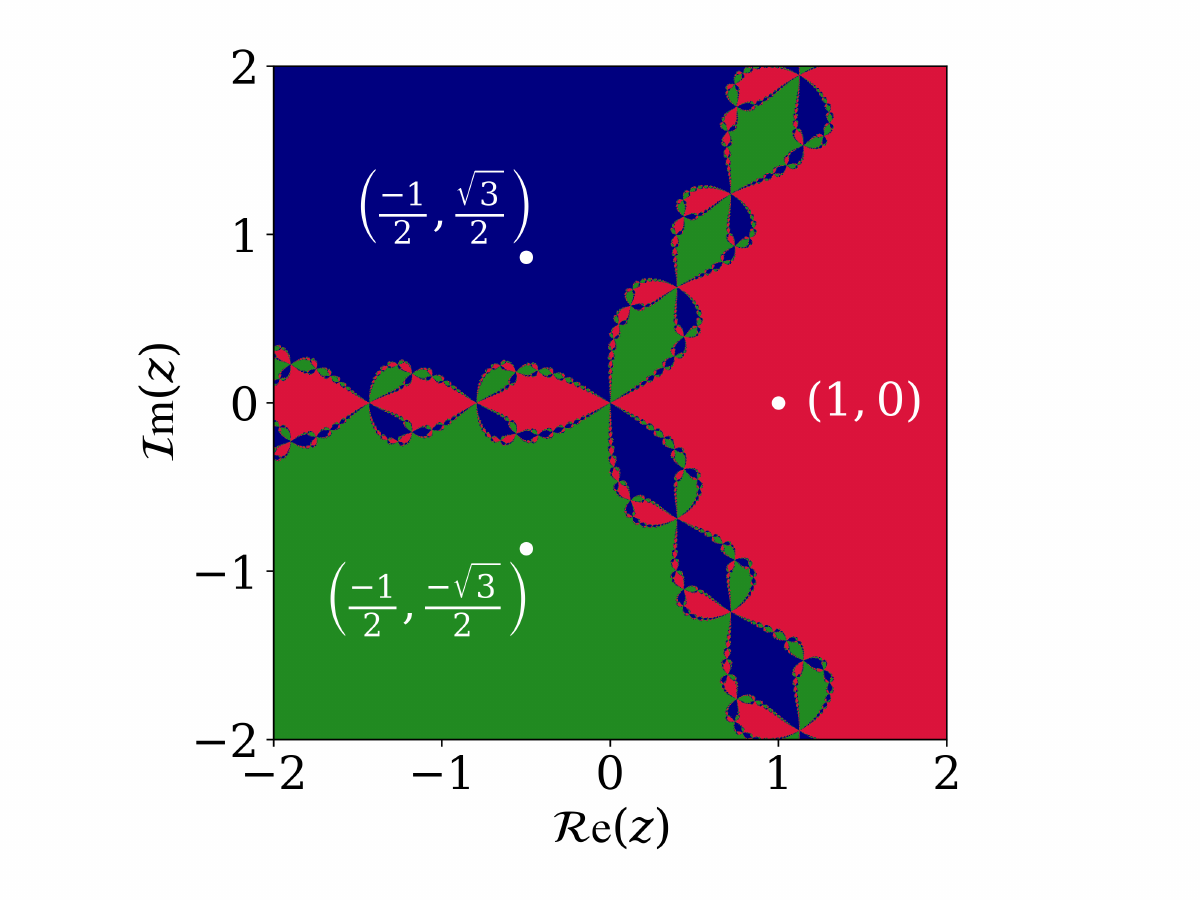}
    \label{fig:Newton_fractal}
     \end{subfigure}
     \hfill
     \begin{subfigure}[b]{0.495\textwidth}
    \centering
    \includegraphics[width = \textwidth]{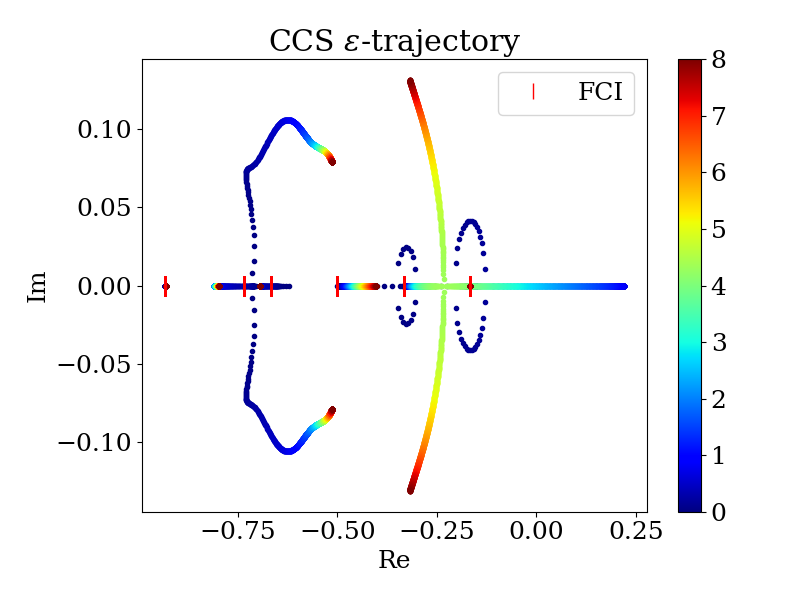}
    \label{fig:strong_correlation}
     \end{subfigure}
      \caption{\label{fig:Motivation}
      Left: Newton fractal of $p(z) = z^3 -1$. The white dots correspond to the roots $z_1, z_{2}, z_{3}$. The colored regions, red, blue, and green, correspond to the basins of attraction of the roots $z_1$, $z_2,z_3$, respectively~\cite{faulstich2023homotopy}. 
      Right: Energy trajectory of CCS solutions for a two-electron system. The overlap of the eigenstates with the reference state is steered by the parameter $\varepsilon$. The plot shows the $\varepsilon$-energy trajectory of all CCS solutions, where $\varepsilon$ was varied between zero and eight. For more details see~\cite{faulstich2022coupled}.}
\end{figure}

This shows that around the individual roots Newton's method convergence towards the closest root. However, it also shows that the global convergence behavior of (quasi) Newton-type methods is highly complicated~\cite{hubbard2001find, schleicher2002number}.
One can only imagine how intricate the Newton fractal of the high-dimensional CC equations is.
These considerations raise the pressing question:\\

\noindent
\begin{center}
Which CC root has been approximated, and is this the ``best'' solution attainable with the considered CC method?
\end{center}

To definitively answer this question one must leave the perturbative framework, theoretically as well as practically!
Mathematically, the most promising framework for studying systems of polynomial equations is algebraic geometry. This field not only provides a set of advanced theoretical tools but also has seen a tremendous surge in computational advances. 
Exploiting parallel implementations, computational procedures (mostly) based on the homotopy continuation method, e.g., PHCpack~\cite{verschelde1999algorithm}, Bertini~\cite{bates2006bertini}, HOM4PS~\cite{chen2014hom4ps,lee2008hom4ps}, NAG4M2~\cite{bates2023numerical}, and HomotopyContinuation.jl~\cite{breiding2018homotopycontinuation}, provide a reasonable starting point to numercally investigate the intricate root structures of the high-dimensional and non-linear CC equations.

Within the chemistry community, the root structure of the CC equations has been studied at a fundamental level with the goal of including homotopy continuation methods in the CC methodology.
The first study on this topic dates back to 1978 when Živkovič and Monkhorst investigated the singularities and multiple solutions of the equations~\cite{vzivkovic1978analytic}.
This was followed by mathematical and numerical studies of multiple solutions of the single-reference and state-universal multi-reference CC equations and their singularities and analytic properties in the early 1990s by Paldus and coworkers~\cite{piecuch1990coupled,paldus1993application}. 
In 1998, Kowalski and Jankowski revived the homotopy methods in connection with the CC theory and used them to solve the CC equations with doubles for a minimum-basis-set four-electron problem~\cite{kowalski1998towards}.
This was followed by a fruitful collaboration of Kowalski and Piecuch, who extended the application of the homotopy methods
to the equations defining the CC approaches with singles and doubles (CCSD), singles, doubles,
and triples (CCSDT), and singles, doubles, triples, and quadruples (CCSDTQ)~\cite{piecuch2000search}, again using
a four-electron system described by a minimum basis set as a target. 
They also introduced the formalism of $\beta$-nested equations and proved the {\it Fundamental Theorem of the $\beta$-NE Formalism}, which enabled them to explain the behavior of the curves connecting multiple solutions of the various CC polynomial systems, i.e., from CCSD to CCSDT, CCSDT to CCSDTQ, etc. 
In~\cite{kowalski2000complete2}, Piecuch and Kowalski used homotopy methods to determine all solutions of nonlinear state-universal multireference CCSD equations based on the Jeziorski-Monkhorst ansatz, proving two theorems that provided an explanation for the observed intruder solution problem. 
In a sequel work~\cite{kowalski2000complete}, they used homotopy methods to obtain all solutions of the generalized Bloch equation, which is nonlinear even in a CI parametrization.\\
\indent
Despite these intensive investigations, the practical computational use of this approach has been restricted to only very small model systems, primarily because of two key reasons. Firstly, to effectively integrate computational algebraic methods with cutting-edge computational quantum chemistry, a substantial scientific divide must be bridged, one that involves advanced and abstract mathematical principles. Secondly, in the late 1980s and 1990s, the field of computational nonlinear algebra was in its infancy, presenting a pioneering yet challenging academic environment for advancements.\\

Recently, a novel computational shift adopting a fully algebraic geometry perspective of CC theory was established~\cite{faulstich2022coupled,faulstich2023algebraic}. This approach has demonstrated significant potential in reshaping our understanding of the CC theory~\cite{faulstich2022coupled,faulstich2023algebraic,faulstich2023homotopy,borovik2023coupled}. In preliminary works, the authors Faulstich, Oster, Strumfels, and Sverrisdóttir have demonstrated that the CC equations possess rich mathematical structures. By integrating these structures into the computational model, the authors were able to significantly reduce the computational scaling of {\it algebro computational} methods applied to the CC equations allowing the computation of all CC roots for small molecular systems~\cite{faulstich2023algebraic}.

The following chapter is outlined as follows. We begin with a brief review of the fundamental concepts underlying the homotopy continuation method in Sec.~\ref{sec:h-cont}. We then discuss different bounds to the number of roots to the CC equations and introduce the crucial concept of truncation varieties in Sec.~\ref{sec:boundingRoots}. In Section~\ref{sec:AlgbRootsNumnerics}, we review the essential numerical discoveries yielded by this approach, providing a detailed analysis of its implications.

\subsection{Homotopy continuation}
\label{sec:h-cont}

Most {\it algebro computational} methods are built on the idea of homotopy continuation -- the numerical approach established in~\cite{faulstich2023algebraic} is no exception.
The idea of homotopy continuation is simple: continuously transform a simple system of polynomials with known solutions into a more complex one and track the paths of these solutions. More formally, we consider the CC equations, written in the following form
\begin{equation}
\label{eq:TargetSystem}
f_{\rm CC}({\bf t})
=
\left[\begin{array}{c}
f_1({\bf t}) \\
\vdots \\
f_m({\bf t})
\end{array}\right]
=
\left[\begin{array}{c}
f_1(t_1, ..., t_m) \\
\vdots \\
f_m(t_1, ..., t_m)
\end{array}\right]
=
0.
\end{equation}

This is our target system, i.e., the system we wish to solve.
In a general case, we require the number of equations to be larger than the number of variables, however, the CC equations are a square system, i.e., we have as many equations as variables.
In order to find all roots to the system in Eq.~\eqref{eq:TargetSystem}, we construct an auxiliary system of polynomial equations denoted $g(\mathbf{t})=0$.
For the construction of this system, two fundamental criteria must be met: firstly, its roots of $g$ should be known, and secondly, the sytem $g$ must have at least as many roots as the target system $f_{\rm CC}$. While meeting the first condition is relatively simple, the second condition poses a greater challenge, as accurately determining the number of roots in the CC equations is a hard problem, see Sec.~\ref{sec:boundingRoots}. Having $f_{\rm CC}$ and $g$, we define a family of systems $H(\mathbf{t}, \lambda)$ for $\lambda \in \mathbb{R}$ interpolating between $f_{\rm CC}$ and $g$, i.e., $H(\mathbf{t}, 0)=f_{\rm CC}(\mathbf{t})$ and $H(\mathbf{t}, 1)=g(\mathbf{t})$. 
For the sake of illustration, we now consider one root $\mathbf{s}_0$ of $g$ and restrict $\lambda \in[0,1]$. 
The condition $H(\mathbf{t}, \lambda)=0$ then defines a solution path $\mathbf{t}(\lambda) \subset \mathbb{C}^m$ such that $H(\mathbf{t}(\lambda), \lambda)=0$ for $\lambda \in[0,1]$ and $\mathbf{t}(1)=\mathbf{s}_0$. Numerically, this path is followed from $\lambda=1$ to $\lambda=0$ in order to compute one solution $\mathbf{t}_0=\mathbf{t}(0)$ to the target system $f_{\rm CC}$. 
This procedure is equivalent to solving the initial value problem
\begin{equation*}
\frac{\partial}{\partial \mathbf{t}} H(\mathbf{t}, \lambda)\left(\frac{\mathrm{d}}{\mathrm{d} \lambda} \mathbf{t}(\lambda)\right)+\frac{\partial}{\partial \lambda} H(\mathbf{t}, \lambda)=0, \quad \mathbf{t}(1)=\mathbf{s}_0,
\end{equation*}
which is known as the Davidenko differential equation~\cite{davidenko1953new,davidenko1953approximate}. We say that $\mathbf{t}(1)=\mathbf{s}_0$ gets tracked towards $\mathbf{t}(0)$. For this to work, $\mathbf{t}(\lambda)$ must be a regular zero of $H(\mathbf{t}, \lambda)=0$ for every $\lambda \in(0,1]$. In the case of nonregular solutions at $\lambda=0$ endgames are employed which are special numerical methods~\cite{morgan1992computing}.

\begin{figure}[h!]
    \centering
    \includegraphics[width = 0.45 \textwidth]{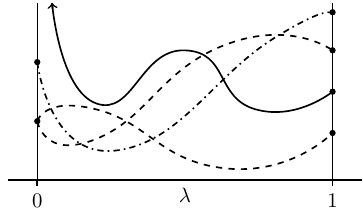}
    \caption{A sketch of possible homotopy paths. The solid line shows a path with no finite limit as $\lambda \to 0$, the dashed lines have the same limit, and the dotted-dashed line has a unique limit.}
    \label{fig:homotopies}
\end{figure}

In analyzing the solution paths traced by the homotopy, as illustrated in Fig.~\ref{fig:homotopies}, various scenarios may arise~\cite{bates2023numerical}. 
One path, represented by the solid line, diverges to infinity as $\lambda \to 0$. 
In contrast, the other three paths converge to finite limits. 
The path indicated by a dotted-dashed line uniquely converges to a regular zero of the target system $f_{\rm CC}$ at $\lambda = 0$. Meanwhile, the two paths denoted by dashed lines converge to a common limit, corresponding to an isolated zero of $f_{\rm CC}$ with a multiplicity of two.
Mathematically, homotopy continuation methods are well studied, we refer the interested reader to~\cite{bates2023numerical,garcia1979finding,morgan2009solving,sommese2005numerical}, and for a quantum chemistry perspective see~\cite{faulstich2023homotopy}.

\subsection{Bounding the number of roots}
\label{sec:boundingRoots}

As becomes apparent from Sec.~\ref{sec:h-cont}, knowing the precise count of roots, or at least a close upper bound, is crucial for the effective application of homotopy methods. This number dictates the number of roots in the auxiliary system $g$ and therewith determines the number of paths to be numerically tracked. Due to the high dimensionality, this turns out to be particularly challenging in the case of CC theory. Subsequently, we denote 
\begin{equation}
{\rm CCdeg}_{N,N_B}(\sigma)
\end{equation}
the true number of roots to the CC equations for a system of $N$ electrons discretized in $N_B$ spin orbitals imposing the CC truncation level $\sigma$, where e.g. $\sigma = \{1,2\}$ stands for CCSD, $\sigma = \{2\}$ stands for CCD, $\sigma = \{1,2,3\}$ stands for CCSDT, etc.

In order to establish a bound to the number of roots to the CC equations~\eqref{eq:ccEqs}, one can start with the simplest estimate for the number of roots in a polynomial system, namely, the Bézout number. The Bézout number is simply the product of the degrees of the individual polynomial equations. In the case of CCSD, this yields 
\begin{equation}
\label{eq:BezoutB}
{\rm CCdeg}_{N,N_B}(\{1,2\})
\leq 
3^{n_s}4^{n_d}
\end{equation}
where $n_s = N(N_B-N)$ is the number of singles equations and $n_d = (N-1)N(N_B-N-1)(N_B-N)$ is the number of doubles equations, see e.g.~\cite{piecuch2000search}.
The Bézout number often greatly overestimates the actual number of roots, as seen in the CC equations~\cite{faulstich2022coupled}. 
For the effective use of homotopy methods, however, it is essential to have precise and accurate estimates of the number of roots. 

One potential way to improve this bound is by means of the Bernstein-Khovanskii-Kushnirenko (BKK) theorem~\cite{Bernshtein1975,Kouchnirenko1976,cox2006using}. The BKK theorem provides a way to estimate the maximum number of solutions that a system of polynomial equations can have, based on the geometric properties of the equations' coefficients. More precisely, it states that for a system of polynomial equations, the number of isolated solutions in the complex domain is bounded by the mixed volume of the Newton polytopes corresponding to the polynomials.
In order to apply this theorem to CC theory, one must investigate the CC Newton polytopes and establish a way to compute or at least bound their mixed volume. This direction was explored in~\cite{faulstich2022coupled}.

Another auspicious direction is the use of {\it truncation varieties}.
This provides significantly improved bounds to the number of CC roots, see~\cite{faulstich2023algebraic} and Fig.~\ref{fig:CC_bounds_bernd}.
The truncation varieties are algebraic varieties specific to CC theory. In general, an algebraic variety is a set of solutions to one or more algebraic equations, typically defined in a higher-dimensional space, where these solutions form a geometric shape or structure. In the context of CC theory, there are several varieties, that appear. Consider the exponential parametrization
\begin{equation}
\label{eq:CCVarietyEq}
\exp~:~\mathcal{V} \to \mathcal{H}_{\rm int}
~;~
{\bf t} \mapsto 
|\Psi\rangle
=
\exp({\bf t})|\Psi_0 \rangle
=
|\Psi_0 \rangle 
+ 
\sum_{n=1}^N \frac{1}{n!} T^n |\Psi_0 \rangle,
\end{equation}
where $\mathcal{V}$ denotes the vector space of CC amplitudes.
Note that Eq.~\eqref{eq:CCVarietyEq} defines a set of algebraic equations. 
Imposing a certain level of truncation corresponds to restricting this map to a subspace of amplitudes $\mathcal{V}_\sigma \subseteq \mathcal{V}$, where $\sigma$ denotes the respective level of truncation as defined above.
We define the truncation variety $V_\sigma$ as the closure of the image of the exponential map of $\mathcal{V}_\sigma$. 
Since the exponential parametrization is invertible, the dimension of the variety $V_\sigma$ is the dimension of $\mathcal{V}_\sigma$.
The truncation varieties exhibit numerous mathematical properties, as elucidated in~\cite{faulstich2023algebraic}, which collectively lead to the bound
\begin{equation}
\label{eq:CCbound}
 {\rm CCdeg}_{N,N_B}(\sigma) 
 \leq 
\bigl({\rm dim}(V_\sigma) + 1 \bigr)  \, {\rm deg}(V_\sigma), 
\end{equation}
where ${\rm deg}(V_\sigma)$ is the degree of the truncation variety $V_\sigma$, which is an intrinsic quantity providing information about the variety's geometric and algebraic properties. In general, the degree of a variety in algebraic geometry refers to a measure of its complexity. It is typically defined as the number of intersections that the variety has with a general linear space of complementary dimension. In simpler terms, it is the number of points at which a linear space will intersect the variety, assuming it intersects it in the maximum possible number of points~\cite{cox2013ideals}. Computing the exact degree for a given truncation variety -- or at least a sufficiently good bound to it -- is the subject of current investigations.

Undertaking a formal comparison between the bounds presented in Eq.~\ref{eq:BezoutB} and Eq.~\eqref{eq:CCbound} is challenging given the fundamentally distinct nature of the underlying concepts involved. Despite this, a numerical comparison reveals that the bound in Eq.~\eqref{eq:CCbound} provides a significant improvement over the previously established bounds, see Sec.~\ref{sec:AlgbRootsNumnerics}.

\subsection{Numerical results}
\label{sec:AlgbRootsNumnerics}

We begin this section by taking a closer look at the convergence behavior of (quasi) Newton-type methods applied to the CC equations~\eqref{eq:ccEqs}. 
Since the dimensionality of the amplitude space grows rapidly, it is not possible to visualize the corresponding Newton fractal. 
However, we can obtain an idea of the size of the basin of attraction around one root, i.e., a ball around one solution within (quasi) Newton-type methods commonly converge to the solution at its center. 
To that end, we consider a variant of the H$_4$ model consisting of four hydrogen atoms symmetrically distributed on a circle of radius 
$R = 1.738$~\AA~\cite{van2000benchmark,bulik2015can} discretized in the STO-3G basis set, see Fig.~\ref{fig:H4_circ_mod}.

\begin{figure}[ht!]
    \centering
    \includegraphics{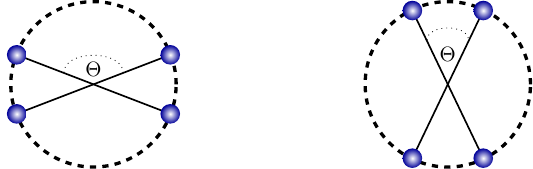}
    \caption{Depiction of the H$_4$ model undergoing a symmetric disturbance on a circle modeled by the angle  $\Theta$.}
    \label{fig:H4_circ_mod}
\end{figure}

For $\Theta = 90^\circ$ we obtain a CC solution ${\bf t}_0$ by initializing the (quasi) Newton-type method with zero. Adding a random perturbation ${\bf t}_p$ to this solution provides a different initialization ${\bf t}_{\rm init} = {\bf t}_0 + {\bf t}_p$ for the CC computations. Scaling the size of ${\bf t}_p$ (i.e., $\Vert {\bf t}_p \Vert$,) allows us to (approximately) investigate the basin of attraction. Clearly, comparing with Fig.~\ref{fig:Motivation}, we expect that the region for which Newton's method converges to ${\bf t}_0$ will not be circular. 
However, this investigation yields the ballpark for the local basin of convergence, since we can extract the radius of the largest ball $r_{\rm max}$ in which (quasi) Newton-type methods converge to the solution ${\bf t}_0$ in $99.9\%$ of the cases, see Fig.~\ref{fig:CCBasinOfAtt}. 
We moreover plot the success rate of Newton's method, i.e., how many of the randomly perturbed initializations converged toward ${\bf t}_0$, as a function of the size of ${\bf t}_p$. Note that we measure $\Vert {\bf t}_p\Vert$ relative to the size of ${\bf t}_0$, in particular, the initialization zero lies on the boundary of $\Vert {\bf t}_p \Vert =1$, see Fig.~\ref{fig:CCBasinOfAtt}.

\begin{figure}[ht!]
     \begin{subfigure}[t]{0.495\textwidth}
    \includegraphics[width = \textwidth]{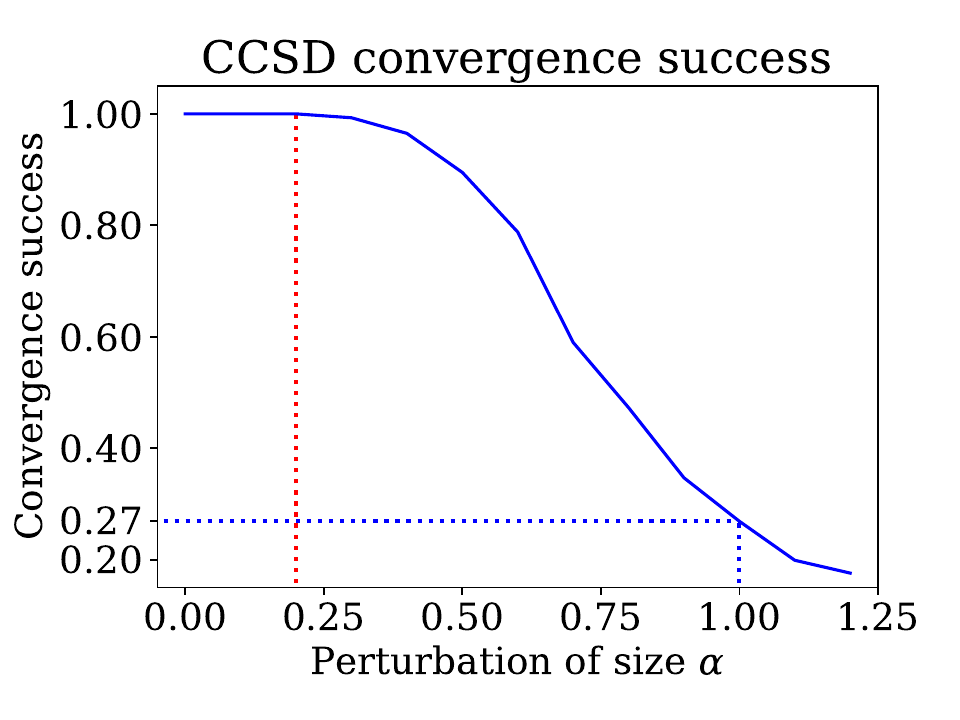}
     \end{subfigure}
     \hfill
     \begin{subfigure}[t]{0.45\textwidth}
    \includegraphics[width = \textwidth]{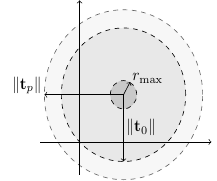}
     \end{subfigure}
      \caption{\label{fig:CCBasinOfAtt}
      Left: The convergence success of 1000 simulations initializing CCSD with the optimal solution ${\bf t}_0$ plus a random perturbation ${\bf t}_p$ of size $\alpha = \Vert {\bf t}_p \Vert/\Vert {\bf t}_0\Vert$. The red dotted line indicates $r_{\rm max}$. The blue dotted line indicates the success rate within a ball around ${\bf t}_0$ of radius $\Vert {\bf t}_0\Vert$.\\
      Right: Schematic representation of the different regions of convergences around ${\bf t}_0$.}
\end{figure}

This shows that $r_{\rm max}$ is approximately $0.2\; \Vert {\bf t}_0 \Vert $. Moreover, this shows that beyond this point convergence towards ${\bf t}_0$ is by no means guaranteed. In fact, for an arbitrary initialization that is $\Vert {\bf t}_0 \Vert$ away from ${\bf t}_0$, the success rate is only $27\%$. Being oblivious about the physical motivation of this initial guess, one could argue that it is quite surprising that Newton's method converges for the initial guess zero. 

We now compare the new bound to the CC roots derived in~\cite{faulstich2023algebraic} with the existing bounds reported in e.g.~\cite{piecuch2000search}.
To that end, we compute the roots corresponding to CCS and CCD for two-electron systems, i.e., $N=2$, for different numbers of spin orbitals $N_B$. This shows that using the truncation varieties and their profound mathematical structures dramatically improved the bounds to the CC roots, see Fig.~\ref{fig:CC_bounds_bernd}.

\begin{figure}[ht!]
    \centering\includegraphics[width = \textwidth]{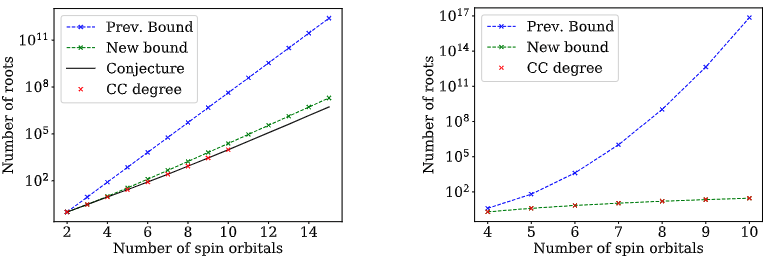}
    \caption{Bounds to the number of roots for 2 electron systems for a varying number of spin orbitals using the CCS method (left) and CCD method (right). For more details see~\cite{faulstich2023algebraic}.}
    \label{fig:CC_bounds_bernd}
\end{figure}

This reduction in the bounds together with the incorporation of truncation varieties in the computational procedure allowed for severe numerical advancements enabling the computation of the full root structure for true molecular systems like lithium hydride (see Fig.~\ref{fig:LiH_energies})~\cite{faulstich2023algebraic} using CCD.

\begin{figure}[ht!]
    \centering
    \includegraphics[width = 0.5 \textwidth]{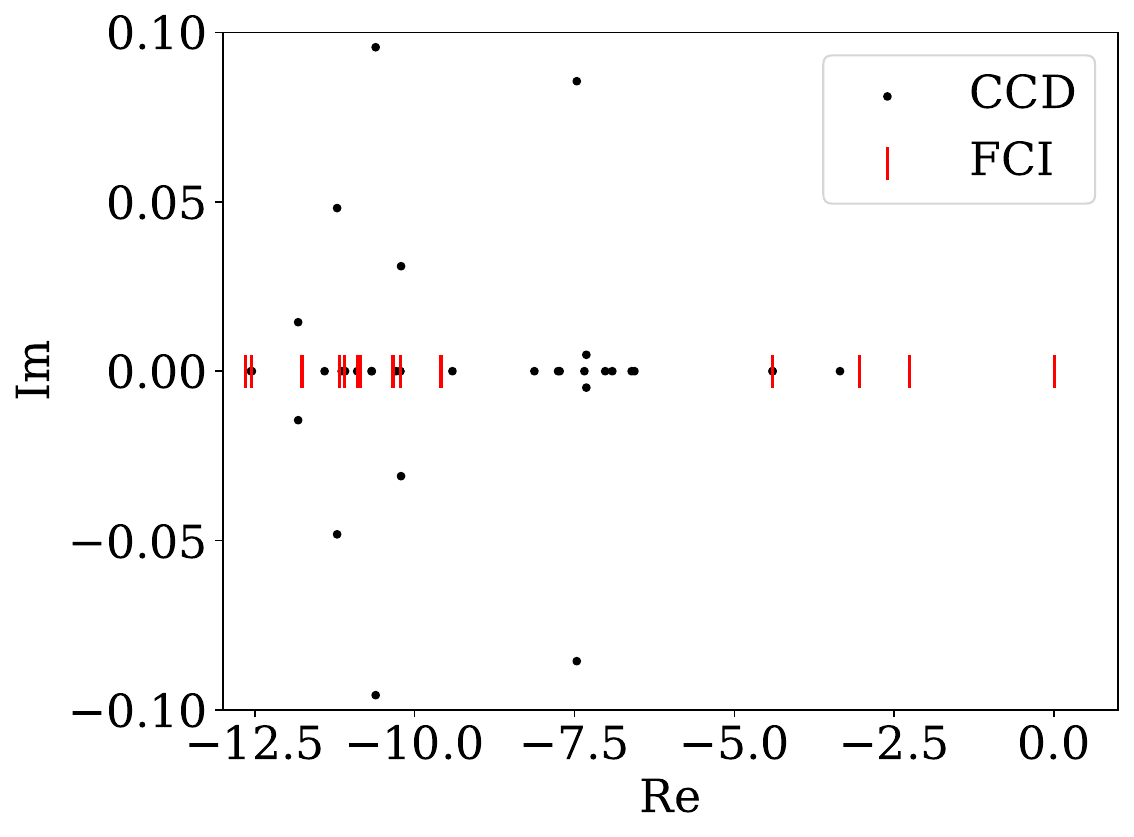}
    \caption{Comparison of energy spectra obtained via FCI and CCD for lithium hydride~\cite{faulstich2023algebraic}.}
    \label{fig:LiH_energies}
\end{figure}

We emphasize that these advances are far from a straightforward application of off-the-shelf computational algebra tools. 
Instead, they result from a sophisticated combination of multiple techniques, underscoring the complexity and innovation of the approach.
The general computational procedure comprises two major steps:

1. The set-up of an initial system from which the homotopy continuation starts. This initial system is specific for the number of electrons, the number of spin orbitals employed for the discretion of the Hamiltonian, and the used CC truncation level $\sigma$ as defined in Sec.~\ref{sec:boundingRoots}. We emphasize that in this implementation, the initial system can be reused when computing CC solutions at the truncation level $\sigma$ for systems with the same number of electrons and basis functions. 

2. Once the initial system for a target system configuration is set up, we employ a parametric homotopy approach as described in Sec.~\ref{sec:h-cont} that connects the initial system with the targeted system.

\section{Conclusion}

This article provides a self-contained educational review of the latest mathematical developments in coupled cluster (CC) theory from a computational chemistry perspective. 
To that end, we started this review article with a foundational introduction to CC theory, employing an algebraic approach. This particular formulation offers a rigorous and mathematically elegant framework, thereby facilitating a deeper understanding of the underlying principles. Additionally, in an effort to ensure comprehensive coverage and to augment the article's self-contained nature, we have incorporated a detailed analysis of the matrix structures that emerge within the realm of second quantization. This includes an exploration of their theoretical underpinnings and practical implications in computational chemistry, providing valuable context and enhancing the overall utility of this review for researchers in the field.
\\

We then explore a variety of analytical frameworks and methods used in CC theory, with a focus on their contributions to establishing local existence and uniqueness of the CC solutions. We delve into the local analysis based on Zarantonello's Lemma, a technique pioneered by Schneider~\cite{schneider2009analysis}, which has significantly influenced the field by its application in various CC methods, including the continuous single-reference CC method~\cite{rohwedder2013continuous, rohwedder2013error}, the extended CC method~\cite{laestadius2018analysis}, and the tailored CC ansatz~\cite{faulstich2019analysis}.
Further, we explore the graph-based framework for CC methods developed by Csirik and Laestadius~\cite{csirik2023coupled1,csirik2023coupled2}. This section highlights the versatility of the framework and its utility in comparing various CC methods, encompassing even multireference approaches.
We then delved into the latest numerical analysis results analyzing the single reference CC method developed by Hassan, Maday, and Wang. This segment decodes the complex ansatz from a computational chemistry viewpoint and encapsulates key findings from their research presented in~\cite{hassan2023analysis,hassan2023analysis2}, offering readers a comprehensive understanding of this cutting-edge area in CC theory.\\

Furthermore, our review extends to the algebraic geometry approach within CC theory. This unique perspective not only illuminates the intricate root structure inherent in CC equations but also paves the way for novel computational paradigms. These emerging methodologies have the potential to form the cornerstone of future CC computational strategies. In our discussion, we delve into the overarching principles of the algebraic approach and incorporate an overview of the most recent numerical advancements that have been made in this area~\cite{faulstich2019analysis,faulstich2023algebraic}.

\section*{Acknowledgments}

The author is thankful for useful discussions with Andre Leastadius, Mihály Csirik, Muhammad Hassan, and Svala Sverrisdóttir.

\newpage
\bibliographystyle{IEEEtranS}
\bibliography{lib.bib}

\end{document}